\documentclass{article}
\usepackage[margin=2.5cm]{geometry}

\newcommand\blfootnote[1]{%
  \begingroup
  \renewcommand\thefootnote{}\footnote
{#1}%
  \addtocounter{footnote}{-1}%
  \endgroup
}

\title{\vspace{-1.0cm}Reconciling the accuracy-diversity trade-off in recommendations}
\author{
  Kenny Peng, Manish Raghavan, Emma Pierson, Jon Kleinberg\textsuperscript{\textdagger}, Nikhil Garg\textsuperscript{\textdagger}\blfootnote{Kenny Peng, Cornell Tech, \url{kennypeng@cs.cornell.edu}; Manish Raghavan, MIT; Emma Pierson, Cornell Tech; Jon Kleinberg, Cornell University; Nikhil Garg, Cornell Tech, \url{ngarg@cornell.edu}. \textsuperscript{\textdagger}Joint senior author.}
}
\date{}

\usepackage{amsmath}
\usepackage{amssymb}
\usepackage{amsthm}
\usepackage{hyperref} 
\usepackage{cleveref}
\usepackage{graphicx}
\usepackage{bm}
\usepackage{booktabs}
\usepackage{dsfont}
\usepackage{mathtools}
\usepackage{tikz}
\usepackage{appendix}
\usepackage{pgfplots}
\usepackage{xcolor}
\usepackage{csquotes}
\usepackage{cite}
\usepackage{booktabs}
\usepackage{enumitem}

\newtheorem{proposition}{Proposition}
\newtheorem{theorem}{Theorem}
\newtheorem{corollary}[proposition]{Corollary}
\newtheorem{lemma}{Lemma} %

\theoremstyle{definition}
\newtheorem{example}{Example}
\newtheorem{definition}{Definition}

\usepackage{apptools}
\AtAppendix{\counterwithin{lemma}{section}}

\newtheorem{innercustomthm}{Lemma}
\newenvironment{mylem}[1]
  {\renewcommand\theinnercustomthm{#1}\innercustomthm}
  {\endinnercustomthm}

\newcommand{\EE}{\mathbb{E}}
\newcommand{\ZZ}{\mathbb{Z}_{\ge 0}}
\newcommand{\RR}{\mathbb{R}}
\newcommand{\DD}{\mathcal{D}}

\newcommand{\wh}{\widehat}

\newcommand{\norm}[1]{\left\lVert#1\right\rVert}

\newcommand{\simiid}{\overset{\mathrm{iid}}{\sim}}

\newcommand{\vocab}[1]{\textbf{\textit{#1}}}

\usepackage{array}
\newcolumntype{P}[1]{>{\centering\arraybackslash}p{#1}}
\newcolumntype{C}[1]{>{\centering\let\newline\\\arraybackslash\hspace{0pt}}m{#1}}

\DeclareMathOperator*{\argmax}{arg\,max}

\DeclareMathOperator{\Ber}{Ber}

\DeclareMathOperator{\topp}{top}

\DeclareMathOperator{\betaa}{Beta}
\DeclareMathOperator{\Exp}{Exp}
\DeclareMathOperator{\Pareto}{Pareto}

\newcommand{\citet}[1]{\cite{#1}}

\hypersetup{
    colorlinks,
    linkcolor={red!50!black},
    citecolor={blue!50!black},
    urlcolor={blue!80!black}
}

\begin{document}

\maketitle

\begin{abstract}
    In recommendation settings, there is an apparent trade-off between the goals of \textit{accuracy} (to recommend items a user is most likely to want) and \textit{diversity} (to recommend items representing a range of categories). As such, real-world recommender systems often explicitly incorporate diversity separately from accuracy. This approach, however, leaves a basic question unanswered: Why is there a trade-off in the first place?

    We show how the trade-off can be explained via a user's \textit{consumption constraints}---users typically only consume a few of the items they are recommended. In a stylized model we introduce, objectives that account for this constraint induce diverse recommendations, while objectives that do not account for this constraint induce homogeneous recommendations. This suggests that accuracy and diversity appear misaligned because standard accuracy metrics do not consider consumption constraints. Our model yields precise and interpretable characterizations of diversity in different settings, giving practical insights into the design of diverse recommendations. 
\end{abstract}

\section{Introduction}

A large body of work in recommendations has developed methods to navigate an apparent trade-off between the goals of \textit{accuracy} (to recommend items a user is most likely to want) and \textit{diversity} (to recommend items from a range of categories) \cite{kunaver2017diversity, Adomavicius2012ImprovingAR, Raza2021DeepNN, Isufi2021AccuracydiversityTI, Alexandridis2015AccuracyVN, Hou2020ATE, Alhijawi2023MultifactorRM, Paudel2017FewerFA, He2022DoesUO, FernndezTobas2016AccuracyAD, Gogna2017BalancingAA, FerrariDacrema2021DemonstratingTE, Liu2012SolvingTA, Eskandanian2020UsingSM, Wu2020ServiceRW, Laci2017BeyondAO, Javari2015APM, Kleinberg2023CalibratedRF, Seymen2021ACO, Abdollahpouri2023CalibratedRA, Zhang2008AvoidingMI, Ashkan2015OptimalGD, Zhou2008SolvingTA}. Real-world recommender systems use heuristics to directly incorporate diversity into recommendations \cite{facebook, ebay}, and empirical evidence demonstrates that users prefer diverse recommendations \cite{10.1145/3366423.3380281, S2022DiversityVR, Kim2021CustomerSO, Park2013FromAT}.

A fundamental question remains: Why is there a trade-off in the first place? More specifically, why is ``accuracy'' unaligned with a user's true preference for diversity? Without a principled understanding of the accuracy-diversity trade-off, attempts to diversify recommendations have difficulty moving beyond a heuristic basis---and difficulty articulating what they are accomplishing at a deeper level.

In this work, we introduce and analyze a stylized model of recommendations that helps explain and reconcile the apparent accuracy-diversity trade-off. Our model studies this trade-off through the lens of a user's \textit{consumption constraints}: users typically examine a list of recommendations and use only the top (highest value) options. (A person can watch only one movie in an evening, a recruiter can select only a handful of candidates to interview.)
This lens has strong explanatory power:
\begin{itemize}[leftmargin=*]
    \item When the goal is to maximize the expected value of the \textit{top} recommended items, diverse recommendations are often optimal in our model. By accounting for a user's consumption constraints, optimizing for what they are most likely to want is aligned with recommending a diverse set of items.
    \item Without accounting for consumption constraints, optimal recommendations in our model are homogeneous. This suggests that the trade-off between accuracy and diversity can be explained by commonly-used accuracy metrics not accounting for consumption constraints.

\end{itemize}

A strength of our model is that we can precisely and interpretably analyze the optimal \textit{amount} of diversity in different settings, allowing us to isolate the effect of consumption constraints. This precision also yields practical insights about the role of diversity in recommendations. In Theorem 3, for example, we uncover natural settings where optimal recommendations overrepresent the category of item a user is \textit{least} likely to want---a paradox that appears in grocery stores, where even though customers are less likely to buy ice cream than milk, stores allocate much more space to ice cream.

\subsection{Model} There are $m$ types of items indexed by $[m]=\{1,2,\cdots,m\}.$ A user prefers exactly one type of item, preferring type $t\in[m]$ with probability $p_t$. The value of the $i$-th item of type $t$ is a random variable $X_i^{(t)}$ if the user prefers type $t$ and $0$ otherwise (so its expected value is $p_t\EE[X_i^{(t)}]$). The recommender knows only how $X_i^{(t)}$ are distributed, not their realizations. We refer to $p_t$ as the \vocab{likelihood} of a type and the random variable $X_i^{(t)}$ as a \vocab{conditional item value} (the value of an item conditional on the user preferring the item's type). 

Let $S_{n,k}$ be the set of $n$ items that maximizes the expected total value of the $k$ items with the highest realized values. We call $S_{n,k}$ the \vocab{optimal} set of $n$ recommendations with respect to this objective. (We omit the dependency of $S_{n,k}$ on the other model parameters, as these will be clear from context.)

$S_{n,k}$ arises naturally from an assumption that the user can only use $k$ items, so that they derive value from the $k$ highest value recommended items. So $S_{n,1}$ results from an assumption that the user uses only one item, while $S_{n,n}$ results from an assumption that the user uses \textit{every} recommended item.

$S_{n,n}$ maximizes the expected total value of all of the recommended items, so by linearity of expectation, $S_{n,n}$ contains the $n$ items with the highest individual expected values. 

Meanwhile, $S_{n,1}$ maximizes the expected value of the highest-value item, an objective that---importantly---is not maximized by choosing the individual items with the highest expected values.

\paragraph{Our overarching technical result and interpretation.} 
Our main results (summarized in \Cref{table:mainresults}) show across several settings that:
\begin{quote}
    $S_{n,k}$ is diverse for $k$ fixed and $n$ growing. Meanwhile, $S_{n,n}$ is homogeneous.
\end{quote} 
This technical result has an interpretation explaining the accuracy-diversity trade-off. Objectives that account for a user's consumption constraints naturally induce diversity, while objectives that do not account for consumption constraints can produce homogeneous recommendations. Thus, the observed trade-off is (partly) a consequence of common accuracy metrics not modeling consumption constraints.

One might suggest that our results show homogeneity is desirable when users can use all recommended items. We do not emphasize this interpretation since there are additional reasons outside of our model to incorporate diversity. Thus, our results are best interpreted as showing (1) that diversity arises under \textit{minimal} assumptions, and (2) that standard objectives can (mistakenly) induce homogeneity.

\paragraph{Evaluating diversity.} To evaluate diversity, we consider the representation of each type in $S_{n,k}$. In many of our results, representation falls on an interpretable continuum: from complete diversity (each type represented equally) to proportional diversity (each type represented proportionally to its likelihood $p_t$) to homogeneity (only the highest likelihood type represented). We formalize our measure of diversity in \Cref{sec:diversity}.

\begin{table}
    \begin{center}
    \caption{Our main results broadly show: (1) consumption constraints induce diverse recommendations, and (2) without accounting for consumption constraints, recommendations tend to be homogeneous.}
    \label{table:mainresults}
      \vspace*{3mm}
      \small
\begin{tabular}{@{}>{\raggedright}p{4.3cm}p{4.3cm}p{4.3cm}}
\toprule
\multicolumn{1}{c}{\textbf{Setting}} & \multicolumn{1}{c}{\textbf{With consumption constraints}} & \multicolumn{1}{c}{\textbf{Without consumption constraints}} \\
\midrule
\textbf{Thm. 1.} 
\begin{equation*}
    X_i^{(t)}\sim \DD
\end{equation*} & As $n$ grows large for fixed $k$, $S_{n,k}$ exhibits diversity depending on the tail behavior of $\DD$. \vspace{0.1cm} \newline For non-heavy-tailed $\DD$, $S_{n,k}$ is at least proportionally diverse. & $S_{n,n}$ contains only one type of item.\\ & & \\
\textbf{Thm. 2.}
\begin{equation*}
    X_i^{(t)}\sim \Ber(q_i)
\end{equation*}for $q_1,q_2,\cdots$ decaying by a power law (roughly, $q_i\propto i^{-\alpha}$). & For moderate amounts of decay $(\alpha < 1)$, as $n$ grows large, $S_{n,1}$ represents each item type equally. & For moderate amounts of decay $(\alpha < 1)$, as $n$ grows large, $S_{n,n}$ is less than proportionally diverse.\\
& & \\
\textbf{Thm. 3.}
\begin{equation*}
    X_i^{(t)}\sim \Ber(q_t)
\end{equation*} for $q_1,q_2,\cdots,q_m$. & For large $n$, $S_{n,1}$ represents each item type in proportion to \begin{equation*}
    \frac{1}{\log \frac{1}{1-q_t}},
\end{equation*}
so that items of \textit{lower} success probability are recommended \textit{more}. & $S_{n,n}$ contains only one type of item.\\
\bottomrule
\end{tabular}
\end{center}
\end{table}

\paragraph{Discussion and limitations.} Our model is purposefully stylized and minimal so that we can abstract a small set of ingredients common to a wide set of domains---namely, those where users prefer a certain category of items and a recommender has noisy estimates of user preferences and item values. Importantly, our model captures the consumption constraints of users: when a user is given a set of recommendations, they can typically use only the best few recommendations---a person can watch only one movie in an evening, a recruiter can select only a handful of candidates to interview. 

\textit{Other reasons for diversity.} Importantly, when we refer to \textit{optimal} sets of recommendations, we mean optimal with respect to our stylized optimization problem---maximizing the expected sum of the $k$ highest item values. In real-world contexts, there are many reasons to incorporate diversity that we do not consider here. For example, our model does not consider any explicit preferences users have for diversity. While such a preference may be empirically well-grounded, omitting such a preference makes our conclusions stronger: simply by modeling user consumption constraints, we show that diversity arises naturally in optimal recommendations \textit{even when} our model does not explicitly value diversity.

\textit{Model generalizations.} Real-world settings can differ from our model in natural ways. For example, users can prefer multiple types of items at a time and some items may fall under multiple types. We discuss generalizations in \Cref{sec:generalizations}. Here, the assumptions that users prefer only one type and items each fall under one type allow us to analyze diversity in an interpretable way: characterizing how represented each item type is in comparison to the user's likelihood of preferring that type.

\subsection{An illustrative example: Recovering Steck's standard of calibration}

We now consider a simple instantiation of our model, motivated by a thought experiment suggested by the Netflix researcher Harald Steck \cite{steck2018calibrated}. A user watches romance movies 70 percent of the time and action movies the other 30 percent of the time. Steck raises a concern that an accuracy-maximizing algorithm can produce entirely homogeneous recommendations in this setting, and proposes a standard of \textit{calibration}, where 70 percent of recommended movies here are romance and 30 percent are action. The example below reflects Steck's concern by showing that---before accounting for consumption constraints---optimal recommendations are homogeneous. Yet, by assuming a user can only watch one movie, the optimal recommendations in the example are in fact naturally \textit{calibrated}.

\begin{example}[Recovering calibration]\label{ex:steck}
    Suppose there are two genres of movies, romance and action, indexed 1 and 2 respectively. A user prefers romance with probability $p_1$ and action with probability $p_2$, with $p_1 > p_2.$ Movies from the user's preferred genre have values drawn i.i.d. from an exponential distribution $\Exp(\lambda)$, while other movies have value $0$. In the language of our model,
    $X_i^{(t)} \simiid \Exp(\lambda)$.
    
    $S_{n,n}$ maximizes the expected total value of recommended items, so by linearity of expectation it contains the items of highest individual expected value. The expected value of each romance and action movie are $p_1 \EE[\Exp(\lambda)]$ and $p_2 \EE[\Exp(\lambda)]$ respectively, so $S_{n,n}$ contains \textit{only} romance movies.

    $S_{n,1}$ maximizes the expected value of the best recommended movie. The expected value of the best item among $a_1$ romance movies and $a_2$ action movies (where $a_1+a_2=n$) is
    \begin{equation}
        p_1\cdot \EE[\text{max of $a_1$ draws from $\Exp(\lambda)$}] + p_2\cdot \EE[\text{max of $a_2$ draws from $\Exp(\lambda)$}]
    \end{equation}
    \begin{equation}
        \approx \frac{p_1}{\lambda} \log(a_1) + \frac{p_2}{\lambda} \log(a_2),
    \end{equation}
     which is maximized when $a_1 = p_1n$ and $a_2 = p_2n$ (as shown using Lagrange multipliers). Therefore, $S_{n,1}$ has proportional representation, recovering Steck's standard of calibration.
\end{example}

\subsection{Intuition for our overarching result}\label{sec:intuition}

In the model we propose, a user's value depends only on the recommended items from the user's preferred type. Therefore, conditional on a user preferring type $t$, the user's expected value from a set of recommendations with $a_t$ items of type $t$ is given by $h_t(a_t)$ for some function $h_t$. Then the user's expected value from a set of recommendations with $a_t$ items of type $t$ for all $t\in [m]$ is of the form
\begin{equation}\label{eq:water}
    \sum_{t=1}^m p_t h_t(a_t).
\end{equation}
The key idea of our work is that when accounting for consumption constraints, the function $h_t$ reflects diminishing returns: recommending additional items from a type becomes less and less valuable when we care only about the value of the best recommended items---when a user can only watch one movie, recommending three options for romance movies might be significantly better than two, because the platform doesn't know which exact romance movie the user may like; however, recommending 20 such movies is barely better than 19. Thus, when maximizing \eqref{eq:water}, it becomes preferable to recommend items from other types. Meanwhile, without modeling user consumption constraints, there are not necessarily diminishing returns since the user can use all of the additional items recommended.

The particular shape of the diminishing returns regulates the amount of resulting diversity. In \Cref{ex:steck}, we had that $h_t(x)\propto \log x,$ in which case \eqref{eq:water} is maximized when $a_t\propto p_t$, yielding proportional representation.\footnote{The mathematical result in this example also appears in the context of resource allocation (e.g., \cite{Ghorbanzadeh2014AUP}) and betting (e.g., \S 22.2 in \cite{easley2010networks}) in the presence of logarithmic utility.} Our technical work thus involves analyzing the functions $h_t$ in different settings (this reduces to analyzing the large order statistics of different distributions). Roughly speaking, heavier-tailed conditional item values imply larger marginal returns, resulting in less diversity.

\subsection{Related work}
As we have noted, there is a wide literature devoted to developing methods to navigate the accuracy-diversity trade-off \cite{kunaver2017diversity, Adomavicius2012ImprovingAR, Raza2021DeepNN, Isufi2021AccuracydiversityTI, Alexandridis2015AccuracyVN, Hou2020ATE, Alhijawi2023MultifactorRM, Paudel2017FewerFA, He2022DoesUO, FernndezTobas2016AccuracyAD, Gogna2017BalancingAA, FerrariDacrema2021DemonstratingTE, Liu2012SolvingTA, Eskandanian2020UsingSM, Wu2020ServiceRW, Laci2017BeyondAO, Javari2015APM, Kleinberg2023CalibratedRF, Seymen2021ACO, Abdollahpouri2023CalibratedRA, Zhang2008AvoidingMI, Ashkan2015OptimalGD, Zhou2008SolvingTA}. Such work is supported by empirical evidence suggesting that a combination of these two metrics is preferred by users \cite{10.1145/3366423.3380281, S2022DiversityVR, Kim2021CustomerSO, Park2013FromAT}. Of particular interest to us, however, is work that focuses on objectives that \textit{implicitly} correspond to diversity. In the context of web search, optimizing for the probability that the user is shown a satisfactory search result has been associated with diversification \cite{10.1145/1498759.1498766, Radlinski2008LearningDR}, since effective search results must account for different intents of queries (``pandas'' can refer to an animal or a Python package). More recent work has adapted this objective to the recommender system setting as a metric that unifies accuracy and diversity \cite{Parapar2021TowardsUM}. Our work, by explicitly characterizing optimal amounts of diversity, identifies consumption constraints as the underlying reason for why such approaches result in diversity. We discuss additional related work in \Cref{sec:additional-related-work}.

\section{Evaluating diversity}\label{sec:diversity}

We now formalize our approach to evaluating diversity. For a set $S$ of items, we define
\begin{equation}
    r_t(S) := \frac{\text{\# of items in }S\text{ of type }t}{|S|},
\end{equation}
the \textbf{representation} of type $t$ in $S$. Intuitively, a set of recommendations is diverse if all types are well represented. We now define an interpretable family of representations that interpolates between maximum diversity and maximum homogeneity, and which arises naturally in many of our results.

\begin{definition}[$\gamma$-homogeneity]
A set $S$ is \textbf{$\bm{\gamma}$-homogeneous} if for all $t\in [m],$
\begin{equation}\label{eq:gamma}
r_t(S) = \frac{p_t^{\gamma}}{\sum_{i=1}^m p_i^{\gamma}}.
\end{equation}
\end{definition}

\noindent $\gamma$-homogeneity captures several intuitive notions of diversity, using $p_1,\cdots,p_m$ as a benchmark:
\begin{itemize}[leftmargin=*]
\item When $\gamma = 0$, $r_t(S)=\frac{1}{m}.$ There is ``equal representation.''
\item When $\gamma = 1$, $r_t(S)=p_t.$ There is ``proportional representation,'' where an item type is represented in proportion to its likelihood.
\item When $\gamma = \infty$, $r_t(S)=1$ for $t=\argmax_{i\in [m]} p_i$ and $r_t(S)=0$ otherwise. There is ``complete homogeneity,'' where only the highest-likelihood item type is represented.
\end{itemize}
A smaller $\gamma$ corresponds to more diversity, with $\gamma\le 1$ indicating \textit{at least proportional} representation. In practice, it is challenging to show that individual sets are $\gamma$-homogeneous; for one, since sets have an integer number of items from each type, it is typically impossible to obtain the exact ratios in \eqref{eq:gamma}. Instead, we will give primarily asymptotic results, showing that as $n$ grows large, the optimal set $S_{n,k}$ approaches $\gamma$-homogeneity. Formally, we define $\gamma$-homogeneity over sequences of sets:

\begin{definition}[$\gamma$-homogeneity for set sequences]
A sequence of sets $\{S_n\}_{n=1}^\infty$ is \textbf{$\bm{\gamma}$-homogeneous} if for all $t\in [m],$
\begin{equation}
\lim_{n\rightarrow \infty} r_t(S_n) = \frac{p_t^{\gamma}}{\sum_{i=1}^m p_i^{\gamma}}.
\end{equation}
\end{definition}

One perhaps surprising aspect of our results is that $\gamma$-homogeneity is sufficient to characterize diversity in a large class of settings, as opposed to requiring more complicated functions of proportions $p_t$. 

\section{Main results}
We now state our main results, which consider several settings reflecting different assumptions about the conditional item values $X_i^{(t)}.$ In \Cref{sec:general}, we assume that $X_i^{(t)}\simiid \DD$ are drawn from a shared distribution, which implies that the recommender has little information about the value of specific items. In \Cref{sec:bernoulli}, $X_i^{(t)}$ are Bernoulli random variables with success probabilities differing depending on $i$ and $t$, meaning that the recommender has information about which items are more likely to satisfy a user.

In each setting, we analyze the diversity of optimal sets $S_{n,k}$ for when $k$ is fixed (i.e., accounting for consumption constraints) and $k=n$ (i.e., not accounting for consumption constraints). We sketch our proof strategy in \Cref{sec:proof-sketch} and defer full proofs to \Cref{sec:proofs}.

\subsection{i.i.d. conditional item values}\label{sec:general}
Consider the setting in which
$X_i^{(t)}\simiid \DD,$
i.e., conditional item values are drawn from a shared distribution. This implies that the values of items behave similarly across types and within types, and the platform cannot easily distinguish between the items in a type. From the perspective of a hiring platform, there may be many candidates with similar backgrounds (e.g., education or work history), none of whom can be distinguished from another by the platform. Conditional on a recruiter preferring this background, candidate values can be modeled as coming from a shared distribution.

We show that for a fixed $k$, as $n$ grows large, the diversity of $S_{n,k}$ theoretically varies between equal representation, proportional representation, and near-complete homogeneity depending on the tail-behavior of $\DD$. In particular, distributions that are bounded or have exponential tails induce at least proportional representation. Meanwhile, $S_{n,n}$ is completely homogeneous.

\begin{theorem}\label{thm:general}
Suppose $X_i^{(t)}\simiid \DD$ where $\DD$ has finite mean. Then the following statements hold.
\begin{enumerate}
\item[(i)] \textbf{[Finite Discrete]} If $\DD$ is a finite discrete distribution, $\{S_{n,k}\}_{n=1}^\infty$ is $0$-homogeneous.
\item[(ii)] \textbf{[Bounded]} If $\DD$ has support bounded from above by $M$ with pdf $f_\DD$ satisfying
\begin{equation}
\lim_{x\rightarrow M} \frac{f_\DD(x)}{(M-x)^{\beta-1}} = c
\end{equation}
for some $\beta, c>0$, then $\{S_{n,k}\}_{n=1}^\infty$ is $\frac{\beta}{\beta+1}$-homogeneous.\\(This pdf class contains beta distributions, including the uniform distribution.)
\item[(iii)] \textbf{[Exponential tail]} If $\DD = \Exp(\lambda)$ for $\lambda > 0,$ then $\{S_{n,k}\}_{n=1}^\infty$ is $1$-homogeneous.
\item[(iv)] \textbf{[Heavy tail]} If $\DD = \Pareto(\alpha)$ for $\alpha>1$, then $\{S_{n,k}\}_{n=1}^\infty$ is $\frac{\alpha}{\alpha-1}$-homogeneous.
\end{enumerate}
Additionally,
\begin{enumerate}
    \item[(v)] $S_{n,n}$ contains only items of type $t = \argmax_{t\in [m]} p_t.$
\end{enumerate}
\end{theorem}

As \Cref{table:thmgeneral} illustrates, the theorem shows how for fixed $k$, the diversity of optimal solutions depends on the tail behavior of $\DD$. In fact, we can obtain $\gamma$-homogeneity for any $\gamma$:
\begin{corollary}
    For any $\gamma\ge 0,$ there exists $\DD$ such that when $X_i^{(t)}\simiid \DD$ and $k$ is fixed, $\{S_{n,k}\}_{n=1}^\infty$ is $\gamma$-homogeneous.
\end{corollary}
Intuitively, heavy-tailed distributions (part (iv)) induce less diverse recommendations since the marginal returns of recommending more items from the same type remains high: drawing more samples from a heavy-tailed distribution produces ever-increasing item values. This contrasts with bounded distributions like the uniform distribution (part (ii)), where once an item has close to the maximum value, additional draws of that type will not further improve the utility significantly. 
\begin{table}\label{table:thmgeneral}
    \begin{center}
    \caption{A summary of Theorem 1. For $X_i^{(t)}\simiid \DD$, distributions $\DD$ with heavier tails induce less diversity.}
    \label{table:thmgeneral}
      \vspace*{3mm}
      \begin{tabular}{l c c c c c}
        \toprule %
        \textbf{} & \multicolumn{3}{c}{\textbf{bounded}} & \textbf{exp. tail} & \textbf{heavy tail}\\
        & \multicolumn{3}{c}{Thm. 1(ii)} & Thm. 1(iii) & Thm. 1(iv)\\
        \cmidrule(r){2-4}
        \cmidrule(r){5-5}
        \cmidrule(r){6-6}
        example $\DD$ &  & $\betaa(\cdot,\beta)$ &  & $\Exp(\lambda)$ & $\Pareto(\alpha)$\\
         & \small{$0<\beta<1$} & \small{$\beta=1$} & \small{$\beta>1$} & \small{$\lambda>0$} & \small{$\alpha>1$}\\        
\includegraphics{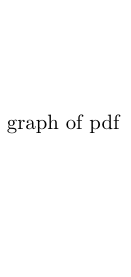}
&\includegraphics{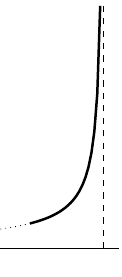}

& \includegraphics{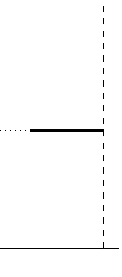}

& \includegraphics{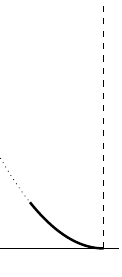} 

&\includegraphics{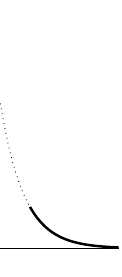}

&\includegraphics{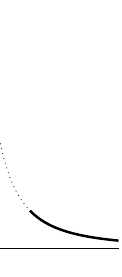}
\\
$\{S_{n,k}\}_{n=1}^\infty$\\$\gamma$-homog.\\for $\gamma\in$ & $(0,1/2)$ & $1/2$ & $(1/2,1)$ & $1$ & $(1,\infty)$\\
&&&&(i.e., proportional)&\\
&&&&&\\
&\multicolumn{5}{c}{$\longleftarrow$ more diverse \qquad \qquad \qquad \qquad \qquad less diverse $\longrightarrow$}\\
\midrule
\end{tabular}
\end{center}
\end{table}

\paragraph{A result for finite $n$ and larger $k$.} One limitation of our main results is that they are asymptotic $(n\rightarrow \infty)$ and are restricted to fixed consumption constraints $k$. Stronger results can be obtained by considering specific distributions. For example, the result below characterizes for any $n, k$ the representation of each type when conditional item values are uniformly distributed on $[0,1].$

\begin{proposition}\label{prop:uniform}
When $X_i^{(t)}\simiid U([0,1])$,
\begin{equation}
    \left|r_t(S_{n,k}) - \frac{\sqrt{p_t}}{\sum_{i=1}^m \sqrt{p_i}}\right|\le \frac{m+1}{n}.
\end{equation}
for all $k\le \frac{\sqrt{p_m}}{\sum_{i=1}^m \sqrt{p_i}}n - m - 1.$
\end{proposition}
Therefore, for any $n$, $S_{n,k}$ is approximately $\frac{1}{2}$-homogeneous. In addition, for any $k$ that is smaller than a constant fraction of $n$, the diversity $S_{n,k}$ does not depend on $k$. Thus, even for small sets of recommendations and large consumption constraints, diversity is optimal in this setting. We further give simulated results for small $n$ and large $k$ in \Cref{sec:simulations} corroborating our theoretical results.
\subsection{Heterogeneous Bernoulli conditional item values}\label{sec:bernoulli}
Now consider when $X_i^{(t)}$ are independent random variables drawn from $\Ber(q_i^{(t)})$, reflecting a model in which items have binary values (i.e., a user is either satisfied or not satisfied by an item). In this section, we allow $q_i^{(t)}$ to differ across $i$ and $t$, implying that the recommender has knowledge about which items are more likely to be successful conditional on a user's preferred type. Specifically, in \Cref{sec:ber-decay} we allow $q_i^{(t)}$ to vary across $i$ and in \Cref{sec:ber-varying} we allow $q_i^{(t)}$ to vary across $t$.

Our results will focus on $S_{n,1}$ and $S_{n,n}$, which both have natural interpretations in this setting:
\begin{itemize}[leftmargin=*]
    \item $S_{n,1}$ maximizes the probability that the user will be satisfied by at least one recommended item.
    \item $S_{n,n}$ maximizes the the expected number of recommended items the user will be satisfied by, which is equivalent to the standard metric of accuracy.
\end{itemize}
Before proceeding, we note that the basic case $q_i^{(t)} = q$ for all $i,t$ is handled as a direct corollary of \Cref{thm:general}(i).

\begin{corollary}[Conditional item values are i.i.d. Bernoulli]\label{cor:ber}
When $X_i^{(t)}\simiid \Ber(q)$ for $q>0,$ then $S_{n,1}$ is $0$-homogeneous.
\end{corollary}

Therefore, if the success probability is the same for all items, optimal solutions are $0$-homogeneous (each item is equally represented) for large $n$, even as the likelihoods $p_t$ vary across type.

\subsubsection{Decaying success probabilities}\label{sec:ber-decay} We now consider a setting in which among items of the same type, the recommender knows that some items have higher success probability. This maps onto settings where the recommender knows which items of a type are most likely to be satisfactory, e.g., some action movies are more commonly liked than are others. Thus, we assume that the recommender has access to items with decaying success probabilities.

\begin{theorem}[Decaying success probabilities]\label{thm:ber-decay}
Suppose that $X_i^{(t)}\simiid \Ber(q_i^{(t)})$ are i.i.d. Bernoulli random variables such that $q_i^{(t)}= c(i+d)^{-\alpha}$ for all $i\ge 1$ and some $\alpha,c,d\ge 0.$ Then the following statements hold.
\begin{enumerate}
    \item[(i)] $\{S_{n,1}\}_{n=1}^\infty$ is $0$-homogeneous for $\alpha < 1.$
    \item[(ii)] $\{S_{n,1}\}_{n=1}^\infty$ is $\frac{1}{1+c}$-homogeneous for $\alpha = 1.$
    \item[(iii)] $\{S_{n,1}\}_{n=1}^\infty$ is $\frac{1}{\alpha}$-homogeneous for $\alpha > 1.$
\end{enumerate}
Additionally,
\begin{enumerate}
    \item[(iv)] $\{S_{n,n}\}_{n=1}^\infty$ is $\frac{1}{\alpha}$-homogeneous for $\alpha\ge 0.$
\end{enumerate}
\end{theorem}

\begin{figure}
\caption{\Cref{thm:ber-decay} considers items decay in quality within each type. Here, we plot optimal diversity as a function of the rate of decay $\alpha$. For realistic rates of decay we ($\alpha < 1$), $S_{n,k}$ is completely diverse for large $n$ while $S_{n,n}$ is less than proportionally diverse.}
\begin{center}
\includegraphics{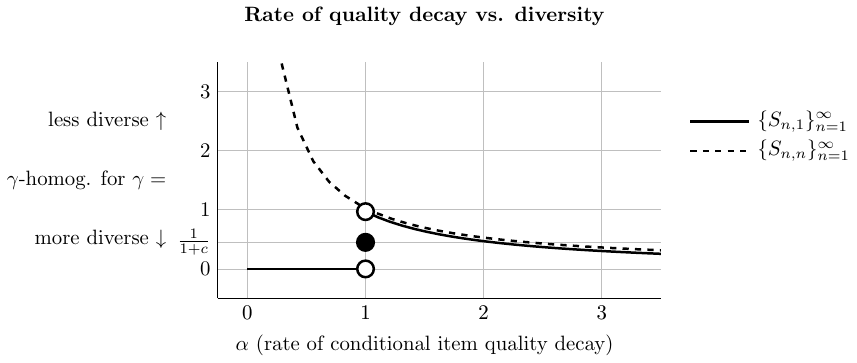}
\end{center}

\end{figure}

When the success probabilities of items have moderate decay ($\alpha<1$), then $0$-homogeneity is maintained in the case $k=1$ (note that $\alpha = 0$ recovers \Cref{cor:ber}). Moreover, for all rates of decay, optimal recommendations reflect at least proportional diversity for large $n$ and $k=1$.

\Cref{thm:ber-decay} also reveals surprising \textit{non-monotonic} behavior. In particular, there is a discontinuity at $\alpha=1$, where homogeneity suddenly increases, but then decreases as $\alpha$ continues to increase. At $\alpha=1,$ the optimal amount of diversity when $\alpha=1$ can range between $0$ and $1$ depending on $c$.

When $k=n$, a larger rate of decay induces more diverse recommendations. Intuitively, when there is a larger rate of decay, the recommender has fewer high-quality options of a given type and is more incentivized to recommend high-quality options of other types. Note that for moderate rates of decay ($\alpha < 1$), $S_{n,n}$ remains less than proportionally diverse for large $n$, unlike $S_{n,1}$.

\subsubsection{Varying success probability across types}\label{sec:ber-varying} We now consider a setting in which the success probability of an item varies across types. This can arise when a users are more picky for some types of items, or when the recommender has more information about items from one type than another.%

\begin{theorem}[Varying success probability across types]\label{thm:ber-varying}
Suppose that for each fixed $t$, $X_i^{(t)} \simiid \Ber(q_t)$ are i.i.d. Bernoulli random variables. Then
\begin{equation}
\lim_{n\rightarrow\infty}r_t(S_{n,1}) \propto \frac{1}{\log \frac{1}{1-q_t}}
\end{equation}
while $S_{n,n}$ contains only items of type $t = \argmax_{t\in [m]} p_tq_t.$
\end{theorem}

The surprising high-level takeaway from \Cref{thm:ber-varying} is that, for large $n$, a \textit{smaller} success probability $q_t$ results in \textit{more} representation of type $t$. The less likely an item of a given type is satisfactory, the more that type is recommended. Moreover, note that the amount of representation in this setting is independent of the popularities $p_1,\cdots,p_m.$

This paradox is illustrated in grocery stores, where more space is devoted to ice cream than milk, despite milk being much more popular than ice cream. Here, $p_1$ (the popularity of milk) is higher than $p_2$ (the popularity of ice cream). However, $q_1$ (the likelihood a given milk product satisfies a shopper looking for milk) is also higher than $q_2$ (the likelihood a given ice cream product satisfies a shopper looking for ice cream), since people tend to have more specific tastes for ice cream. Thus, since $q_2$ is smaller than $q_1$ and the grocery store should ``recommend'' many more ice creams than milks, explaining why more space is devoted to ice cream. Intuitively, while more shoppers want milk, these consumers can be satisfied with a small selection of milk; thus, it is more beneficial to devote more space to ice cream, for which shoppers have more specific tastes.
\section{Sketch of Proofs}\label{sec:proof-sketch}
We now sketch the proof of Theorem 1(ii). (The general strategy applies to the rest of Theorem 1, as well as Theorem 2 and Theorem 3.) First, recall the general setup in Theorem 1, where the conditional item values $X_i^{(t)}$ are drawn i.i.d. from a shared distribution $\DD$. In other words, a user prefers type $t\in [m]$ with probability $p_t$ such that the user prefers exactly one type of item, and conditioned on the user preferring type $t$, the value of an item of that type is drawn i.i.d. from $\DD$. We are interested in analyzing, depending on $\DD$, the composition of $S_{n,k}$, the set of $n$ items maximizing the expected sum of the $k$ highest value items in the set. If $S_{n,k}$ contains $a_t^{(n)}$ items of type $t$, we need to analyze
\begin{equation}\label{eq:lim}
\lim_{n\rightarrow \infty} r_t(S_{n,k}) = \lim_{n\rightarrow \infty} \frac{a_t^{(n)}}{n}.
\end{equation}
We first provide an expression for the expected sum of the $k$ highest value items in a set $S$ with $a_t$ items of type $t$. The following definition will be useful.
\begin{definition}
    Define $\mu_\DD(i,a)$ to be the expected value of the $i$-th order statistic\footnote{The $i$-th order statistic of $n$ random variables is the $i$-th smallest of the $n$ realized values.} of $a$ random variables drawn i.i.d. from $\DD$. (Thus, $\mu_\DD(1,a)$ is the expected minimum of $a$ i.i.d. draws from $\DD$ and $\mu_\DD(a,a)$ is the expected maximum.)
\end{definition}
Then conditioned on the user preferring type $t$, the expected sum of the $k$ highest value of items in $S$ is equal to
   $h(a_t) := \sum_{i=1}^{\min\{k, a_t\}} \mu_\DD(a_t-i+1,a_t),$
which follows from the linearity of expectation. Therefore, the expected sum of the $k$ highest value items in $S$ is $\sum_{t=1}^m p_t h(a_t)$. Define $A_n$ to be the set of tuples of non-negative integers whose entries sum to $n$. Then
$(a_1^{(n)}, a_2^{(n)}, \cdots, a_m^{(n)}) = \argmax_{(a_1,\cdots,a_m)\in A_n} \sum_{t=1}^m p_t h(a_t).$

We can then determine the limit in \eqref{eq:lim} given asymptotic information about $h$. In \Cref{lem:fennel} in the appendix, we develop general technical machinery for this task. Below, we state \Cref{lem:fennel}(ii), which can be used to prove Theorem 1(ii).
\begin{mylem}{\ref{lem:fennel}(ii)}
If $h$ is monotonically increasing and there exist constants $A,B>0$ and $\sigma<0$ such that $\lim_{a\rightarrow \infty} \frac{A-h(a)}{Ba^{\sigma}} = 1,$ then
$\lim_{n\rightarrow \infty} \frac{a_t^{(n)}}{n} = \frac{p_t^{\frac{1}{1-\sigma}}}{\sum_{i=1}^m p_i^{\frac{1}{1-\sigma}}}.$
\end{mylem}
Then, considering $\DD$ as in Theorem 1(ii), we can prove the necessary asymptotic result about $h$:
\begin{lemma}\label{lem:bob}
If $\DD$ has support bounded from above by $M$ with pdf $f_\DD$ such that
$\lim_{x\rightarrow M} \frac{f_\DD(x)}{(M-x)^{\beta-1}} = c$
for some $\beta, c>0$, then
$\lim_{a\rightarrow \infty} \frac{Mk - h(a)}{Ba^{-\frac{1}{\beta}}} = 1.$
\end{lemma}
Combining \Cref{lem:bob} with \Cref{lem:fennel}(ii), with $\sigma = -\frac{1}{\beta}$, we show that for $\DD$ as in Theorem 1(ii),
\begin{equation}
    \lim_{n\rightarrow \infty}\frac{a_t^{(n)}}{n} = \frac{p_t^{\frac{\beta}{\beta+1}}}{\sum_{i=1}^m p_i^{\frac{\beta}{\beta+1}}}.
\end{equation}

\section{Conclusion}
We introduce and analyze a stylized model that reconciles the apparent accuracy-diversity trade-off in recommendations. We characterize the diversity of optimal sets both when ``optimality'' captures and does not capture user consumption constraints. Broadly speaking, we show that the former naturally induces diversity while the latter results in homogeneity. Therefore, the apparent accuracy-diversity trade-off is partially due to traditional accuracy metrics not accounting for consumption constraints.

\paragraph{Limitations and future work.} A particular strength of our model is that we were able to derive precise and interpretable characterizations of diversity in many settings. One limitation of our work is that many of our results are asymptotic (i.e., $n\rightarrow \infty$). We expect that it is possible to obtain further results in our model for finite $n$. We gave \Cref{prop:uniform} as one such example, and outline additional possible directions in \Cref{sec:simulations}.

The purposefully stylized nature of our model---in which users prefer only one type of item and items fall under only one type---allows for a particularly interpretable evaluation of diversity where we compare the representation of each type in comparison to the likelihood that the user prefers that type. This simple model is expressive enough to admit a wide range of results. Still, it would be interesting to generalize our findings in more complex models of users and items. In \Cref{sec:generalizations}, we discuss these possibilities further, providing some basic results both when users can prefer multiple types of items and when users and items are represented by embeddings.

Finally, there are possible model characteristics beyond consumption constraints that implicitly reward diversity. As discussed in \Cref{sec:intuition}, diversity in our model is a consequence of diminishing returns when recommending additional items of a type. Diminishing returns may arise from other assumptions, such as decaying user attention (e.g., \cite{Kleinberg2023CalibratedRF}). We point those interested in considering these alternate assumptions to \Cref{lem:fennel} in the appendix, in which we abstract the technical relationship between diminishing returns and diversity.

\paragraph{Broader impacts.} Recommender systems have broad societal consequences, particularly in high-stakes settings like employment. Our work does not model many additional reasons for diversity in these settings, including those based on fairness and equity. We emphasize again that our results are intended to convey \textit{minimal} assumptions that induce diversity and illustrate the tendency of certain objectives to produce homogeneous recommendations. Designers of recommender systems, and researchers in the area, should take a broader view when making decisions regarding diversity.

\newpage
{\small
\bibliographystyle{unsrt}
\bibliography{bib}
}

\newpage
\pagebreak
\appendix
\section{Additional related work}\label{sec:additional-related-work}

Our work sits at the intersection of two broad sets of work. On the one hand are arguments that diversity is key to achieving efficiency. On the other are those that cast diversity as in conflict with efficiency or accuracy, but perhaps that diversity should nevertheless be pursued as an axiomatic good. 

Broadly, our work seeks to understand this tension by sharply characterizing the \textit{amount} of diversity in efficient solutions, as a function of key setting characteristics: user utilities and consumption constraints, and uncertainty in the item quality distribution. In particular, our results characterize \textit{in what settings} the intuition regarding diversity being efficient holds, and in what settings they may be in conflict.

\paragraph{The (efficiency) benefits of diversity} The importance of diversity for efficiency is an old idea present across many fields; \citet{page2008difference} synthesizes the conceptual and empirical arguments in support of this principle. \citet{hong2004groups} develop a model in which a randomly selected team of problem solvers outperforms a team of the individually best-performing agents, due to diversity in problem solving perspective (\citet{kleinberg2018team} show that, in some settings, there exist \textit{tests} under which selecting the best-performing agents again becomes optimal). \citet{kleinberg2018selection} show that constraints promoting diversity can improve efficiency when they work to counteract a decision-maker's biases. \citet{10.1145/1498759.1498766} develop an algorithm to diversify search results, to minimize the risk of user dissatisfaction. We are particularly influenced by the work of \citet{steck2018calibrated}, who presents the intuition that recommendations should be \textit{calibrated}: ``When a user has watched, say, 70 romance movies and 30 action movies, then it is reasonable to expect the personalized list of recommended movies to be comprised of about 70\% romance and 30\% action movies as well.'' \citet{guo2021stereotyping} show that collaborative filtering-based recommendations may not be able to effectively show users such a diverse set of content, harming efficiency.

More broadly, researchers studying various combinatorial optimization problems may find it obvious that homogeneous solutions can be sub-optimal; indeed, in classical problems like \textit{maximum coverage}, redundancy is undesirable. 

Our work particularly is intimately connected to the large literature on assortment optimization \cite{kok2007demand,rusmevichientong2010dynamic,Davis2014AssortmentOU,jagabathula2014assortment,rusmevichientong2014assortment,gallego2014constrained,bertsimas2015data,el2021joint,chen2022fair,el2022joint}. That literature also considers consumption-constrained consumer item selections based on an intermediary's recommendations (e.g., that customers picks one item  according to a multinomial choice model). The literature primarily devises \textit{approximation algorithms} to find the optimal recommendation (``assortment'') as a function of the consumer's choice model, platform objective, and the item distribution. In other words, an implicit premise of this literature is that the naive approach of presenting the items with highest individual expected values is sub-optimal, i.e., that optimal assortments are not completely `homogeneous.' On the other hand, optimal assortments are not necessarily diverse; roughly speaking, the results of \citet{el2022joint} imply that a standard assortment approach (Mixed MNL) might produce solutions that are not ``diverse'' enough to satisfy multiple customer types, and so there is benefit to personalize to each type.\footnote{We thank the authors for highlighting this connection to us.}\footnote{Furthermore, as \citet{chen2022fair} recently characterize, standard assortment optimization approaches may be ``unfair'' to items in other ways.} Our work contributes to this literature by (a) examining the implicit premise that optimal assortments are not homogeneous (i.e., when is the naive\footnote{Note that \textit{naive}
is much simpler than the \textit{greedy} approach studied in the literature, which picks items iteratively potentially as a function of previous items picked.} approach sufficient?); and (b) showing the characteristics under which optimal assortments are not diverse. 

\paragraph{Diversity and fairness as a contrast to efficiency and accuracy} On the other hand, many works start with the premise that---although diversity may conflict with efficiency or accuracy---it is an axiomatic good that should be pursued. For example, diversity is often considered to be inherently desirable from a fairness perspective and user satisfaction perspective. As a result, there is a wide body of work devoted to optimizing for various metrics of diversity. A common approach (taken, for example, in \citet{carbonell1998use} and \citet{gimpel2013systematic}) is to consider an objective function that balances a weighted measure of ``accuracy'' or ``relevance'' with a measure of diversity. More recently, \citet{brown2022diversified} consider set recommendation for an agent with adaptive preferences, to ensure that consumption over time is diverse. Numerous metrics for diversity have been proposed---we refer the reader to \citet{kunaver2017diversity} for a survey. Similarly, the fair ranking and recommendation literature (see \citet{patro2022fair} and \citet{zehlike2021fairness} for recent surveys) considers metrics and methods for fairness in such problems. On the other hand, empirical work has demonstrated that such tradeoffs may be small in practice \cite{rodolfa2021empirical}. Such formulations imply that there is a tension between diversity and measures of accuracy.
\section{Computational results for finite $n$ and general $k$}\label{sec:simulations}

\begin{figure}
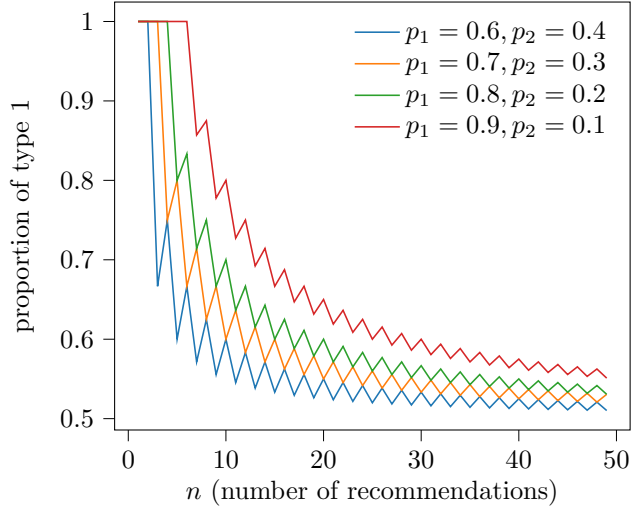

\caption{Optimal recommendations when there are two types (say, comedy movies and action movies) and $X_i^{(t)}$ are i.i.d. Bernoulli random variables with success probability $0.4$. For large $n$, optimal recommendations have equal representation. Here, we plot the optimal representation of type 1 depending on the popularities $p_1$ and $p_2$ for small $n$. When popularities are very skewed (e.g., the user prefers comedy with probability 0.9), then the recommendations are relatively homogeneous for small $n$.}
\label{fig:chart}
\begin{center}
\include{figures/small-n-chart.tex}
\end{center}
\end{figure}

\begin{figure}
\caption{Computational estimates for the diversity of optimal sets across a range of beta and Pareto distributions. Here, dark blue corresponds to equal representation ($\gamma=0$), grey corresponds to proportional representation ($\gamma=1$), and brown corresponds to homogeneous ($\gamma=\infty$). Observe that tighter consumption constraints (smaller $k$) correspond to more diversity. Notice in particular that when $\beta=1$ (corresponding to the uniform distribution), the optimal solution is $\frac12$-homogeneous for all $k\le cn$ for $c\approx \sqrt{0.3}/(\sqrt{0.3}+\sqrt{0.7})\approx 0.4,$ matching our theoretical result in \Cref{prop:uniform}.}
\label{fig:heatmap}
\begin{center}
\includegraphics[width=11cm]{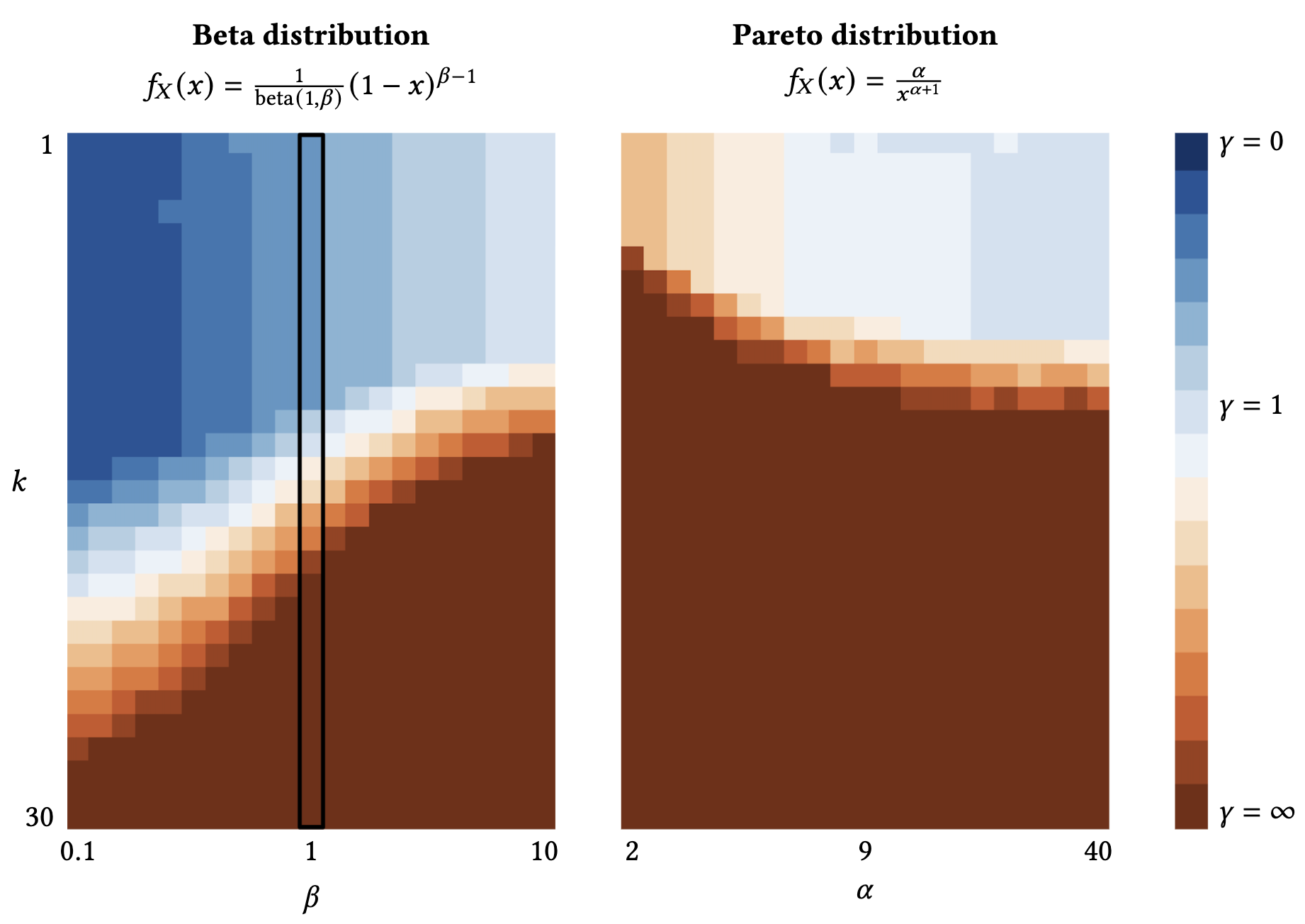}
\end{center}
\end{figure}

One limitation of our general theoretical approach is that it does not directly yield results for a fixed number of recommended items $n$ or for intermediate values of $k$. In particular, a finer-grained theoretical analysis requires closed-form solutions for the order statistics of general distributions, which are often not available. We were able to overcome this in \Cref{prop:uniform}, where we leveraged the simplicity of uniform order statistics.

Here, we take a computational approach to determine the optimal amount of diversity for small $n$ and arbitrary $k.$ \Cref{fig:heatmap} shows the results of a computational approach to calculate the diversity of the optimal set for finite $n$ and general $k$ for a spectrum of Beta and Pareto distributions. For each particular distribution $\DD$, we take $n=30, m=2, p_1=0.7, p_2=0.3$ and find the optimal choice of $a_1$ and $a_2$ by calculating
\begin{equation}
\sum_{t=1}^2 p_t \sum_{i=1}^{\min\{a_t,k\}} \mu_\DD(a_t-i+1,a_t)
\end{equation}
for each possible pair $(a_1,a_2)$ summing to $30$. To do this, we obtained empirical estimates of $\mu_\DD(a-i+1,a)$ by computationally sampling $10^6$ sets of $a$ points from $\DD$. Our computational results corroborate \Cref{thm:general}, showing how the optimal amount of diversity can range from equal representation to homogeneity.

We also compute how diverse optimal solutions are in the case of a Bernoulli distribution as $n$, the number of recommendations, varies. In \Cref{fig:chart}, we consider the distribution $\Ber(0.4)$ and different pairs of popularities $(p_1,p_2).$ We find that as $n$ grows, optimal solutions indeed approach having equal representation, corroborating \Cref{thm:general}(i). For small $n$, optimal solutions tend to be more homogeneous---especially when the popularities are heavily skewed. Indeed, this aligns with intuition: if a user almost exclusively prefers romance, a recommender should still recommend mostly romance movies (even if as the number of recommendations grows very large, the recommendations should become more diverse according to \Cref{thm:general}(i)).
\section{Generalizations}\label{sec:generalizations}
In this section, we consider two directions in which our model can be extended. First, we consider a model in which users can prefer multiple item types at a time. Second, we consider a model in which users and items are both represented by vector embeddings, a common paradigm in recommender systems.

\subsection{Preference for multiple item types}

In this section, we consider an extension of our model in which a user can prefer multiple types of items. In particular, we focus on the case when $m=2$ and when the user prefers only type $1$ with probability $p_1$, only type $2$ with probability $p_2$, and both types with probability $p_{12}.$ We also assume that an item of type $i$ satisfies a user that prefers type $i$ independently with probability $q$. We consider when $k=1$, when the user can only use one item. Then, we would like to minimize the probability of having no satisfying items:
\begin{equation}
    p_1(1-q)^{a_1} + p_2(1-q)^{a_2} + p_{12}(1-q)^n,
\end{equation}
where $a_1 + a_2 = n$.
This is equivalent to minimizing
\begin{equation}
    p_1(1-q)^{a_1} + p_2(1-q)^{a_2},
\end{equation}
which is mathematically equivalent to \Cref{thm:general}(i). Thus, we have that as $n$ grows large, both item types are represented equally, irrespective of $p_1, p_2,$ and $p_{12}$---in other words, allowing for the possibility that the user prefers multiple item types does not change the result in this setting.

The setup here can be naturally generalized to the case in which there are more than two types, and where item values can come from distributions other than Bernoulli. We leave these generalizations to future work.

\subsection{User and item embeddings} In this section, we consider the setting in which users and items each correspond to vector embeddings, and where the likelihood a user $u$ is satisfied by an item $v$ is a function of their cosine distance. Suppose that a user has preference $u$ drawn according to a probability measure $\mu$ supported on the unit $d-$dimensional sphere $S_d\subset \RR^d.$ Further suppose that for a movie $v\in S_d,$ the probability that a user is not satisfied by the movie is $p(u,v) = q(\norm{u-v})$, a function of the cosine distance between $u$ and $v$.

If the goal is to maximize the probability a user is satisfied by at least one item, what is the optimal choice of items to recommend as the number of recommendations we can make grows large? Here, we will show that items should be chosen ``uniformly'' from $S_d$---a result that may seem surprising as it is independent of the distribution of user preferences, and which can be viewed as an analog of \Cref{thm:general}(i).

To make things precise and tractable, rather than consider the recommendation of individual items, we will focus on the recommendation of a ``distribution of items.'' This is analogous to making a continuous relaxation of the discrete item recommendation problem. For a set of items $V = \{v_1,\cdots,v_n\},$ the probability a user with preference drawn according to a probability measure $\mu$ does not like any item in $V$ is
\begin{equation}
    \int_{S_d} \mu(u) \prod_{v\in V} p(u,v) \,du = \int_{S_d} \mu(u) \exp \sum_{v\in V} \log p(u,v) \,du.
\end{equation}
Then for a density function $\alpha: S_d \rightarrow [0,\infty)$, we can consider the expression
\begin{equation}
    \int_{S_d} \mu(u) \exp \left[n\int_{S_d} \alpha(v) \log p(u,v)\,dv\right] \,du,
\end{equation}
which can be thought of splitting the $n$ item recommendations continuously across $S_d$ according to the density function $\alpha$. We can think of $\alpha$ as representing some ``profile'' of items to show. Going forward, we consider optimal distributions rather than optimal discrete sets of items.

\begin{proposition}
Consider a non-constant function $p(u,v): S_d \times S_d \rightarrow (0,1]$, interpreted as the probability that item $v$ does not satisfy a user with preference $u$, such that $p(u,v) = q(\norm{u-v})$ can be expressed as a function of the cosine distance between $u$ and $v$.
Then let $\mu$ be a probability measure with support $S_d$, representing the distribution of user preferences. Then for a probability measure $\alpha$ on $S_d$, define
\begin{equation}
    \Gamma(\alpha) = \lim_{n\rightarrow \infty}\frac{1}{n}\log \int_{S_d} \mu(u) \exp \left[n \int_{S_d} \alpha(v)\log p(u,v)\, dv\right]\,du.
\end{equation}
Then
\begin{equation}
    \Gamma(\pi) \in \inf_{\alpha} \Gamma(\alpha),
\end{equation}
where $\pi$ is the uniform probability measure over $S_d$.
\end{proposition}

\begin{proof}
Let
\begin{equation}
    \rho(u;\alpha) = \int_{S_d} \alpha(v)\log p(u,v)\,dv.
\end{equation}
Then
\begin{equation}
    \Gamma(\alpha) = \lim_{n\rightarrow\infty} \frac{1}{n} \log \int_{S_d} e^{n\rho(u;\alpha)} \mu(u)\, du = \lim_{n\rightarrow\infty} \frac{1}{n} \log \int_{S_d} e^{n\rho(u;\alpha)}\, du = \sup_{u\in S_d} \rho(u;\alpha),
\end{equation}
where the final equality follows from the Laplace principle from large deviations theory. Now note that
\begin{align}
    \int_{S_d} \rho(u;\alpha)\,du &= \int_{S_d} \int_{S_d} \alpha(v) \log p(u,v) \,dv\,du\\
    &= \int_{S_d} \alpha(v) \int_{S_d} \log p(u,v) \,du\,dv \\
    &= \int_{S_d} \alpha(v)C\, dv\\
    &= C,
\end{align}
where the second to last equality follows from the observation that
\begin{equation}
    \int_{S_d} \log p(u,v)\, du = \int_{S_d} \log q(\norm{u-v})\, du = C
\end{equation}
for a constant $C$ independent of $v$.
Therefore, there exists $u$ such that $\rho(u;\alpha) \ge \frac{C}{m(A)},$ so $\sup_{u\in A}\rho(u;\alpha) \ge \frac{C}{m(A)}.$ Now note that when $\pi$ is the uniform probability measure over $A$, we have
\begin{equation}
    \rho(u;\pi) = \int_{S_d} \pi(v) \log p(u,v) \,dv = \frac{C}{m(A)}
\end{equation}
for all $u$, so $\sup_{u\in A}\rho(u;\pi) = \frac{C}{m(A)}.$ So $\Gamma(\pi) = \inf_\alpha \Gamma(\alpha).$
\end{proof}
\newpage
\section{Proofs}\label{sec:proofs}
In this section, we prove each of our results. In \Cref{sec:central-lemma}, we prove a technical lemma that plays a central role in the proofs of each part of \Cref{thm:general} and \Cref{thm:ber-decay}, which we prove in \Cref{sec:proof-thm1} and \Cref{sec:proof-thm2} respectively. In \Cref{sec:proof-rounding-lemma}, we prove another useful lemma---showing that integer solutions are close to real solutions for a class of optimization problems---which will be helpful in the proofs of \Cref{prop:uniform} and \Cref{thm:ber-varying}, which we give in \Cref{sec:proof-of-uniform} and \Cref{sec:proof-thm3} respectively.

\subsection{A central lemma}\label{sec:central-lemma}
Let $\ZZ$ denote the set of non-negative integers and $A_n\subset \ZZ^m$ denote the set of $m$-tuples whose elements sum to $n$. We will say that a function $h: \ZZ\rightarrow \RR$ is \vocab{strictly concave} if $h(a+1)-h(a) < h(a)-h(a-1)$ for all $a$.

\begin{lemma}\label{lem:fennel}
Consider an integer $m$ and $p_1,p_2,\cdots,p_m\ge 0$. Let $h:\ZZ \rightarrow \RR$ be monotonically increasing. For each positive integer $n$, choose $(a_1^{(n)},\cdots,a_m^{(n)})$ such that
 \begin{equation}
        (a_1^{(n)},\cdots,a_m^{(n)}) \in \argmax_{(a_1,\cdots,a_m)\in A_n} \sum_{t=1}^m p_t h(a_t),
    \end{equation}
and define 
\begin{equation}
    r_t^{(n)}:= \frac{a_t^{(n)}}{n}.
\end{equation}
Then the following statements hold.
\begin{enumerate}
        \item[(i)] Suppose there exist constants $A,B>0$ and $\sigma<0$ such that
        \begin{equation}\lim_{a\rightarrow \infty} \frac{\log(A - h(a))}{Ba^\sigma} = 1.\end{equation}
        Then
        \begin{equation}
            \lim_{n\rightarrow \infty} r_t^{(n)} = \frac{1}{m}.
        \end{equation}
        \item[(ii)] Suppose there exist constants $A,B>0$ and $\sigma<0$ such that
        \begin{equation}
            \lim_{a\rightarrow \infty} \frac{A-h(a)}{Ba^{\sigma}} = 1.
        \end{equation}
        Then
        \begin{equation}
            \lim_{n\rightarrow \infty} r_t^{(n)} = \frac{p_t^{\frac{1}{1-\sigma}}}{\sum_{i=1}^m p_i^{\frac{1}{1-\sigma}}}.
        \end{equation}
        \item[(iii)] Suppose $h$ is strictly concave, and that there exist constants $B,C>0$ such that
        \begin{equation}
        \lim_{a\rightarrow\infty} h(a) - B\log a - C = 0.
        \end{equation}
        Then
        \begin{equation}
            \lim_{n\rightarrow \infty} r_t^{(n)} = p_t.
        \end{equation}
        \item[(iv)] Suppose $h$ is strictly concave, and that there exist constants $B>0$ and $0 < \sigma < 1$ such that
        \begin{equation}
            \lim_{a\rightarrow \infty} \frac{h(a)}{Ba^{\sigma}} = 1.
        \end{equation}
        Then
        \begin{equation}
            \lim_{n\rightarrow \infty} r_t^{(n)} = \frac{p_t^{\frac{1}{1-\sigma}}}{\sum_{i=1}^m p_i^{\frac{1}{1-\sigma}}}.
        \end{equation}
    \end{enumerate}
\end{lemma}

We spend the remainder of the section proving \Cref{lem:fennel}. A useful first step is to show that in each of parts (i)-(iv), we have that
\begin{equation}\label{eq:lim-a}
    \lim_{n\rightarrow \infty} a_t^{(n)} = \infty
\end{equation}
for each $t\in [m],$ allowing us to use the asymptotic assumptions in the lemma's statement.

Assume for the sake of contradiction that there exists $t\in [m]$ and an integer $d$ such that for any integer $N$ there exists $n>N$ for which $a_t^{(n)} < d.$ 
Since $h$ is strictly increasing and $d$ is finite, there exists $\delta>0$ such that 
\begin{equation}
    h(a+1)-h(a) > \delta
\end{equation} for all $a < d.$ Also, there exists an integer $N'$ such that for all $a>N'$, 
\begin{equation}\label{eq:sand}
    h(a)-h(a-1) < \delta\cdot \min_{i\in [m]} \frac{p_t}{p_i}. 
\end{equation}
\eqref{eq:sand} holds in parts (i)-(ii) because $h$ is monotonically increasing and is upper bounded by $A$, and in parts (iii)-(iv) because $h$ is strictly concave.

Now consider $N = N'm$. Then there exists $n > N$ such that $a_t^{(n)} < d.$ Since $\sum_{t=1}^m a_t^{(n)} = n > N'm$, there exists $t'\in [m]$ such that $a_{t'}^{(n)} > N'.$ Thus,
\begin{equation}
    p_{t'} h(a_{t'}^{(n)}) - p_{t'} h(a_{t'}^{(n)} - 1) < p_t\delta < p_t h(a_t^{(n)} + 1) - p_t h(a_t^{(n)}),
\end{equation}
which implies that the switch $a_t^{(n)} \rightarrow a_t^{(n)} + 1, a_{t'}^{(n)}\rightarrow a_{t'}^{(n)} - 1$ increases
\begin{equation}
    \sum_{t=1}^m p_t h(a_t^{(n)}),
\end{equation}
contradicting the optimality of $(a_1^{(n)},\cdots, a_m^{(n)}).$

With \eqref{eq:lim-a} in hand, we turn to the bulk of the proof. In each part, we would like to show that
\begin{equation}
\lim_{n\rightarrow \infty} (r_1^{(n)},\cdots,r_m^{(n)}) = (\wh{r}_1, \cdots, \wh{r}_m)    
\end{equation} 
for some specified $(\wh{r}_1, \cdots, \wh{r}_m)$ depending on the part. We assume for the sake of contradiction that $\{(r_1^{(n)},\cdots,r_m^{(n)})\}_{n=1}^\infty$ does not converge to $(\wh{r}_1, \cdots, \wh{r}_m).$ If this is the case, then by the Bolzano-Weierstrass theorem, since $[0,1]^m$ is compact, there is a subsequence
\begin{equation}\{(r_1^{(s_i)},\cdots,r_m^{(s_i)})\}_{i=1}^\infty
\end{equation}
such that
$\lim_{i\rightarrow \infty} (r_1^{(s_i)},\cdots,r_m^{(s_i)}) = (r_1,\cdots,r_m)$
for some $(r_1,\cdots,r_m) \neq (\wh{r}_1, \cdots, \wh{r}_m).$ For notational ease, we will simply assume that
\begin{equation}
\lim_{n\rightarrow \infty} (r_1^{(n)},\cdots,r_m^{(n)}) = (r_1,\cdots,r_m)
\end{equation}
for some $(r_1,\cdots,r_m) \neq (\wh{r}_1, \cdots, \wh{r}_m)$. The proof holds analogously when the subsequence $\{(r_1^{(s_i)},\cdots,r_m^{(s_i)})\}_{i=1}^\infty$ differs from $\{(r_1^{(n)},\cdots,r_m^{(n)})\}_{n=1}^\infty$.

Then consider any sequence $\{(\wh{a}_1^{(n)},\cdots, \wh{a}_m^{(n)})\}_{n=1}^\infty$ such that
\begin{equation}
    \lim_{n\rightarrow \infty} \left(\frac{\wh{a}_1^{(n)}}{n},\cdots, \frac{\wh{a}_m^{(n)}}{n}\right) = (\wh{r}_1, \cdots, \wh{r}_m).
\end{equation}
(Clearly, such a sequence exists.) In each part, we will show that for sufficiently large $n$,
\begin{equation}
    \sum_{t=1}^m p_t h(a_t^{(n)}) < \sum_{t=1}^m p_t h(\wh{a}_t^{(n)}),
\end{equation}
contradicting the optimality of $(a_1^{(n)},\cdots, a_m^{(n)}).$ To complete the proof, we analyze each part separately:

\begin{enumerate}
    \item[(i)] In this part, there exist constants $A,B>0$ and $\sigma<0$ such that
        \begin{equation}\lim_{a\rightarrow \infty} \frac{\log(A - h(a))}{Ba^\sigma} = 1.\end{equation}
        We set $\wh{r}_t := \frac{1}{m}$ for each $t\in [m].$
    Observe that
    \begin{align}
        \lim_{a\rightarrow \infty} \frac{\log (A - h(a))}{Ba^\sigma} = 1
    \end{align}
    implies that for all $\epsilon > 0$, there exists $c$ such that for all $a>c,$
    \begin{equation}
        e^{(1-\epsilon)Ba^\sigma} \le A-h(a)\le e^{(1+\epsilon)Ba^\sigma}.
    \end{equation}
    Then observe that by taking sufficiently small $\epsilon,$ we have
\begin{align}
\lim_{n\rightarrow \infty} \frac{\sum_{t=1}^m p_t (A - h(a_t^{(n)}))}{\sum_{t=1}^m p_t (A - h(\wh{a}_t^{(n)}))} &\ge \lim_{n\rightarrow \infty} \frac{\sum_{t=1}^m p_t \exp[(1-\epsilon)B(a_t^{(n)})^\sigma)]}{\sum_{t=1}^m p_t \exp[(1+\epsilon)B(\wh{a}_t^{(n)})^\sigma)]}\\
&= \lim_{n\rightarrow \infty} \frac{\sum_{t=1}^m p_t \exp[(1-\epsilon)B(nr_t)^\sigma]}{\exp[(1+\epsilon)B(n/m)^\sigma]}\\
&= \lim_{n\rightarrow \infty} \sum_{t=1}^m p_t \exp[Bn^{\sigma}((1-\epsilon)r_t^\sigma - (1+\epsilon)(1/m)^\sigma)] = \infty
\end{align}
where the last limit holds for $\epsilon$ sufficiently small because $r_t - \frac{1}{m} > 0$ for some $t$.

It follows that for $n$ sufficiently large, $\sum_{t=1}^m p_t h(a_t^{(n)}) < \sum_{t=1}^m p_t h(\wh{a}_t^{(n)}),$ as desired.
    
    \item[(ii)] In this part, there exist constants $A,B>0$ and $\sigma<0$ such that
        \begin{equation}
            \lim_{a\rightarrow \infty} \frac{A-h(a)}{Ba^{\sigma}} = 1.
        \end{equation}
        We set
    \begin{equation}
       \wh{r}_t := \frac{p_t^{\frac{1}{1-\sigma}}}{\sum_{i=1}^m p_i^{\frac{1}{1-\sigma}}} 
    \end{equation}
    for each $t\in [m].$
    Then observe that
    \begin{align}
&\lim_{n\rightarrow \infty} \frac{\sum_{t=1}^m p_t  (A - h(a_t^{(n)}))}{\sum_{t=1}^m p_t  (A - h(\wh{a}_t^{(n)}))}\\
= &\lim_{n\rightarrow \infty} \frac{\sum_{t=1}^m p_t  (A - h(a_t^{(n)}))}{\sum_{t=1}^m p_t  (A - h(\wh{a}_t^{(n)}))} \cdot \lim_{n\rightarrow \infty} \frac{\sum_{t=1}^m p_t B(a_t^{(n)})^\sigma}{\sum_{t=1}^m p_t  (A - h(a_t^{(n)}))} \cdot \lim_{n\rightarrow \infty} \frac{\sum_{t=1}^m p_t  (A - h(\wh{a}_t^{(n)}))}{\sum_{t=1}^m p_t B(\wh{a}_t^{(n)})^\sigma} \label{eq:mug-2}\\
= &\lim_{n\rightarrow \infty} \frac{\sum_{t=1}^m p_t  (A - h(a_t^{(n)}))}{\sum_{t=1}^m p_t  (A - h(\wh{a}_t^{(n)}))} \cdot \frac{\sum_{t=1}^m p_t B(a_t^{(n)})^\sigma}{\sum_{t=1}^m p_t  (A - h(a_t^{(n)}))} \cdot \frac{\sum_{t=1}^m p_t  (A - h(\wh{a}_t^{(n)}))}{\sum_{t=1}^m p_t B(\wh{a}_t^{(n)})^\sigma} \label{eq:mug-3}\\
= &\lim_{n\rightarrow \infty} \frac{\sum_{t=1}^m p_t B(a_t^{(n)})^\sigma}{\sum_{t=1}^m p_t B(\wh{a}_t^{(n)})^\sigma} \\
= & \frac{\sum_{t=1}^m p_tr_t^\sigma}{\sum_{t=1}^m p_t\wh{r}_t^\sigma} > 1,\label{eq:mug}
\end{align}
where \eqref{eq:mug-2} follows from the latter two limits being equal to 1, \eqref{eq:mug-3} follows from the product rule for limits, and \eqref{eq:mug} follows from the observation that for $\sigma < 0$
\begin{equation}
    \sum_{t=1}^m p_t x_t^\sigma,
\end{equation}
subject to the constraint $\sum_{t=1}^m x_t = 1$ for $x_t\ge 0$ has a unique minimum at $(x_1,\cdots,x_m)=(\wh{r}_1,\cdots,\wh{r}_m).$ This is direct, for example, by using Lagrange multipliers. \eqref{eq:mug} implies that
\begin{equation}
    \lim_{n\rightarrow \infty} \frac{\sum_{t=1}^m p_t  h(a_t^{(n)})}{\sum_{t=1}^m p_t  h(\wh{a}_t^{(n)})} < 1.
\end{equation}

It follows that for $n$ sufficiently large, $\sum_{t=1}^m p_t h(a_t^{(n)}) < \sum_{t=1}^m p_t h(\wh{a}_t^{(n)}),$ as desired.

    \item[(iii)] In this part, $h$ is strictly concave, and there exist constants $B,C>0$ such that
        \begin{equation}
        \lim_{a\rightarrow\infty} h(a) - B\log a - C = 0.
        \end{equation}
        We set $\wh{r}_t := p_t$ for each $t\in [m].$
    Then observe that
\begin{align}
\lim_{n\rightarrow \infty} \sum_{t=1}^m p_t h(a_t^{(n)}) - \sum_{t=1}^m p_t h(\wh{a}_t^{(n)}) &= \lim_{n\rightarrow \infty} \sum_{t=1}^m p_t B\log a_t^{(n)}-\sum_{t=1}^m p_t B\log \wh{a}_t^{(n)}\\
&= B\log n + B\sum_{t=1}^m p_t \log r_t - B\log n - B\sum_{t=1}^m p_t \log \wh{r}_t\\
&< 0.
\end{align}
The final inequality here follows from the observation that
\begin{equation}
    \sum_{t=1}^m p_t \log x_t,
\end{equation}
subject to the constraint $\sum_{t=1}^m x_t = 1$ for $x_t>0$ has a unique minimum at $(x_1,\cdots,x_m)=(\wh{r}_1,\cdots,\wh{r}_m).$ This is direct, for example, by using Lagrange multipliers.

It follows that for $n$ sufficiently large, $\sum_{t=1}^m p_t h(a_t^{(n)}) < \sum_{t=1}^m p_t h(\wh{a}_t^{(n)}),$ as desired.
    
    \item[(iv)] In this part, $h$ is strictly concave, and there exist constants $B>0$ and $0 < \sigma < 1$ such that
        \begin{equation}
            \lim_{a\rightarrow \infty} \frac{h(a)}{Ba^{\sigma}} = 1.
        \end{equation}
        We set
    \begin{equation}
       \wh{r}_t := \frac{p_t^{\frac{1}{1-\sigma}}}{\sum_{i=1}^m p_i^{\frac{1}{1-\sigma}}} 
    \end{equation}
    for each $t\in [m].$
    Then observe that
\begin{align}
\lim_{n\rightarrow \infty} \frac{\sum_{t=1}^m p_t h(a_t^{(n)})}{\sum_{t=1}^m p_t h(\wh{a}_t^{(n)})} = \lim_{n\rightarrow \infty} \frac{\sum_{t=1}^m p_t B(a_t^{(n)})^\sigma}{\sum_{t=1}^m p_t B(\wh{a}_t^{(n)})^\sigma} = \frac{\sum_{t=1}^m p_tr_t^\sigma}{\sum_{t=1}^m p_t\wh{r}_t^\sigma} < 1.
\end{align}
The first equality is a consequence of the asymptotic assumption on $h$ and the product rule for limits (as in part (ii). The final inequality here follows from the observation that for $\sigma > 0$,
\begin{equation}
    \sum_{t=1}^m p_t x_t^\sigma,
\end{equation}
subject to the constraint $\sum_{t=1}^m x_t = 1$ for $x_t>0$ has a unique maximum at $(x_1,\cdots,x_m)=(\wh{r}_1,\cdots,\wh{r}_m).$ This is direct, for example, by using Lagrange multipliers.

It follows that for $n$ sufficiently large, $\sum_{t=1}^m p_t h(a_t^{(n)}) < \sum_{t=1}^m p_t h(\wh{a}_t^{(n)}),$ as desired.
    
\end{enumerate}

\subsection{Proof of \Cref{thm:general}}\label{sec:proof-thm1}
We now turn to the proofs of \Cref{thm:general}(i)-(iv). In each of these parts, we consider a set of recommendations with $a_t$ items of type $t$ for each $t\in [m].$ Then observe that the expected total value of the $k$ highest value recommended items is equal to
\begin{equation}
    \sum_{t=1}^m p_t h(a_t),
\end{equation}
for
\begin{equation}
    h:\ZZ \rightarrow \RR, \quad h:a\mapsto \EE\left[\topp_k\{X_1,\cdots,X_a\simiid \DD\}\right],
\end{equation}
where $\topp_k$ evaluates the sum of the $k$ highest values in a set. Intuitively, conditional on a user preferring type $t$, the top $k$ items are just the top $k$ items recommended of type $t$. The sum of their values, conditioned on the user preferring type $t$, is simply the sum of the $k$ highest values among $a$ random draws from $\DD$. Clearly, $h$ here is monotonically increasing.

Then, with \Cref{lem:fennel} in hand, parts (i)-(iv) reduces to showing the following:
\begin{enumerate}
    \item[(i)] If $\DD$ is a finite discrete distribution, there exist constants $A,B>0$ and $\sigma>0$ such that  \begin{equation}\lim_{a\rightarrow \infty} \frac{\log(A - h(a))}{Ba^\sigma} = 1.\end{equation}
    \item[(ii)] If $\DD$ has support bounded from above by $M$ with pdf $f_\DD$ satisfying
    \begin{equation}
    \lim_{x\rightarrow M} \frac{f_\DD(x)}{(M-x)^{\beta-1}} = c
    \end{equation}
    for some $\beta, c>0$, then there exist constants $A,B>0$ such that
        \begin{equation} \lim_{a\rightarrow \infty} \frac{A-h(a)}{Ba^{-\frac{1}{\beta}}} = 1.
        \end{equation}

    \item[(iii)] If $\DD = \Exp(\lambda)$ for $\lambda > 0,$ then $h$ is strictly concave and there exists a constant $B>0$ such that
        \begin{equation}
            \lim_{a\rightarrow \infty} \frac{h(a)}{B\log a} = 1.
        \end{equation}
    \item[(iv)] If $\DD = \Pareto(\alpha)$ for $\alpha>1$, then $h$ is strictly concave and there exists a constant $B>0$ such that
        \begin{equation}
            \lim_{a\rightarrow \infty} \frac{h(a)}{Ba^{\frac{1}{\alpha}}} = 1.
        \end{equation}
\end{enumerate}

The following identity, mentioned in \Cref{sec:proof-sketch}, will be useful for parts (ii)-(iv).
\begin{proposition}\label{prop:mu}
For $X_i^{(t)}\simiid \DD$,
    \begin{equation}
    h(a) = \sum_{i=1}^{\min\{k, a\}} \mu_\DD(a-i+1,a).
    \end{equation}
\end{proposition}
Recall that $\mu_\DD(i,a)$ is the expected value of the $i$-th order statistic of $a$ random variables drawn i.i.d. from $\DD$.

\begin{proof}
    Let $Y_{k,n}$ be the $k$-th order statistic of $n$ random variables distributed i.i.d. from $\DD$. Then
\begin{equation}\label{eq:pond}
\topp_{k}\{X_1^{(t)},\cdots,X_{a}^{(t)}\}
= \sum_{i=1}^{\min\{k,a\}} Y_{a-i+1, a}.
\end{equation}
So, as desired,
\begin{equation}
    \EE\left[\topp_{k}\{X_1^{(t)},\cdots,X_{a}^{(t)}\}\right] = \sum_{i=1}^{\min\{k,a\}} \EE\left[Y_{a-i+1,a}\right] = \sum_{i=1}^{\min\{k, a\}} \mu_\DD(a-i+1,a),
\end{equation}
where the first equality follows from \eqref{eq:pond} and the linearity of expectation.
\end{proof}

\paragraph{Proof of \Cref{thm:general}(i).}
Suppose $\DD$ has support $\{x_1,\cdots,x_r\}$ with $x_1>\cdots>x_r$ such that for $X\sim \DD$, $\Pr[X=x_1] = q.$ Now consider a set of recommendations with $a_t$ items of type $t$ for each $t\in [m].$ Then consider $X_1,\cdots,X_a\simiid \DD$. Let $E$ be the event that at least $k$ of $X_1,\cdots,X_a$ equal $x_1$. Then,
\begin{align}
    h(a) &\ge \EE[\topp_k\{X_1,\cdots,X_a\}|E] \cdot \Pr[E]\\
    &= x_1k\cdot \left(1 - \sum_{j=0}^{k-1}\binom{a}{j}(1-q)^{a-j}q^j\right)\\
    &\ge x_1k(1 - a^k(1-q)^{a-k+1})
\end{align}
for all $a>2.$ Now let $E'$ be the event that at least one of $X_1,\cdots, X_a$ equals $x_1$. Then,
\begin{align}
    h(a) &= \EE[\topp_k\{X_1,\cdots,X_a\}|E'] \cdot \Pr[E'] + \EE[\topp_k\{X_1,\cdots,X_a\}|\overline{E'}] \cdot (1-\Pr[E'])\\
    &\le x_1k(1 - (1-q)^{a}) + x_2k(1-q)^{a}\\
    &= x_1k(1 - (1-\frac{x_2}{x_1})(1-q)^{a}).
\end{align}
Now note that for $A = x_1k,$ we have that
\begin{align}
    x_1k(1 - \frac{x_2}{x_1})(1-q)^{a} &\le A - h(a) \le x_1k a^k (1-q)^{a-k+1}\\
    \log(x_1k(1 - \frac{x_2}{x_1})(1-q)^{a}) &\le \log(A - h(a)) \le \log(x_1k a^k (1-q)^{a-k+1})\\
    \log(x_1k) + \log(1 - \frac{x_2}{x_1}) + a\log(1-q) &\le \log(A - h(a)) \le \log(x_1k) + k\log(a) + (a-k+1)\log(1-q).
\end{align}
It follows that for $B = \log(1-q),$
\begin{equation}
    \lim_{a\rightarrow \infty} \frac{\log(A - h(a))}{Ba} = 1,
\end{equation}
as desired. The result follows from \Cref{lem:fennel}(i).

\paragraph{Proof of \Cref{thm:general}(ii).}

First recall from \Cref{prop:mu} that
\begin{equation}
    h(a) = \sum_{i=1}^{\min\{k, a\}} \mu_\DD(a-i+1,a).
\end{equation}
We will show that \begin{equation}
\lim_{a\rightarrow \infty} \frac{Mk - h(a)}{Ba^{-\frac{1}{\beta}}} = 1
\end{equation}
for a constant $B>0.$ \Cref{thm:general}(ii) then follows immediately by applying \Cref{lem:fennel}(ii) with $\sigma = -\frac{1}{\beta}$.

Consider a probability distribution $\DD'$ with pdf $g_X(x)=f_X(M-x)$ and cdf $G_X(x).$ Then
\begin{equation}
\mu_\DD(a-i+1,a) = M - \mu_{\DD'}(i,a),
\end{equation}
which implies that
\begin{equation}
    Mk - \sum_{i=1}^{k} \mu_\DD(a-i+1,a) = \sum_{i=1}^{k} \mu_{\DD'}(i,a)
\end{equation}
Since
\begin{equation}
    \mu_{\DD'}(i,a) = \sum_{j=0}^{i-1} \int_0^\infty \binom{a}{j}G_X(x)^j (1-G_X(x))^{a-j}\,dx,
\end{equation}
it remains to show that for all fixed $j$,
\begin{equation}\label{eq:orangepeel}
    \lim_{a\rightarrow \infty}\frac{\int_0^\infty \binom{a}{j}G_X(x)^j (1-G_X(x))^{a-j}\,dx}{a^{-\frac{1}{\beta}}} = B
\end{equation}
for some constant $B$ (that can vary depending on $j$). Verifying \eqref{eq:orangepeel} comprises the bulk of the technical work of the proof, and we isolate it in the following lemma.

\begin{lemma}
For $\beta>0$,
\begin{equation}
\int_0^\infty \binom{a}{j} G_X(x)^j (1 - G_X(x))^{a-j}\,dx \propto a^{-\frac{1}{\beta}}.
\end{equation}
\end{lemma}

\begin{proof}
We have that
\begin{equation}
    \lim_{x\rightarrow 0^+} \frac{g_X(x)}{cx^{\beta-1}} = \lim_{x\rightarrow M^{-}} \frac{f_X(x)}{c(M-x)^{\beta-1}} = 1
\end{equation}
for a positive constant $c$. So for all $\epsilon>0$ there exists $\delta>0$ such that
\begin{equation}
(1 - \epsilon)cx^{\beta-1} \le g_X(x) \le (1 + \epsilon)cx^{\beta-1}
\end{equation}
for all $x<\delta.$
Now note that $g_X(x) \le (1 + \epsilon)cx^{\beta-1}$ implies that
\begin{equation}
    G_X(x) = \int_0^x g_X(u)\,du \le (1+\epsilon)\int_0^x cu^{\beta-1}\,du = (1+\epsilon)\frac{c}{\beta}x^\beta.
\end{equation}
Likewise, $g_X(x) \ge (1 - \epsilon)cx^{\beta-1}$ implies that
\begin{equation}
    G_X(x) = \int_0^x g_X(u)\,du \ge (1-\epsilon)\int_0^x cu^{\beta-1}\,du = (1-\epsilon)\frac{c}{\beta}x^\beta.
\end{equation}
Now write
\begin{align}
&a^{\frac{1}{\beta}}\int_0^\infty \binom{a}{j} G_X(x)^j (1 - G_X(x))^{a-j}\,dx\\ &=a^{\frac{1}{\beta}}\int_0^\delta \binom{a}{j} G_X(x)^j (1 - G_X(x))^{a-j}\,dx + a^{\frac{1}{\beta}}\int_\delta^\infty \binom{a}{j} G_X(x)^j (1 - G_X(x))^{a-j}\,dx.
\end{align}
We will analyze these two integral separately. It will turn out that the second integral vanishes as $a$ grows.
\end{proof}
\paragraph{The first integral.} We have that
\begin{align}
&a^{\frac{1}{\beta}}\int_0^\delta \binom{a}{j} G_X(x)^j (1 - G_X(x))^{a-j}\,dx\\
&\le a^{\frac{1}{\beta}}\int_0^\delta \binom{a}{j} (1+\epsilon)^j\left(\frac{c}{\beta}\right)^j x^{\beta j} (1 - (1-\epsilon)\frac{c}{\beta}x^\beta)^{a-j}\,dx\\
&= \int_0^{\delta a^{\frac{1}{\beta}}} \binom{a}{j} (1+\epsilon)^j\left(\frac{c}{\beta}\right)^j \left(\frac{x}{a^{\frac{1}{\beta}}}\right)^{\beta j} \left(1 - (1-\epsilon)\frac{c}{\beta}\left(\frac{x}{a^{\frac{1}{\beta}}}\right)^\beta\right)^{a-j}\,dx\\
&= \int_0^{\delta a^{\frac{1}{\beta}}} \binom{a}{j} (1+\epsilon)^j\left(\frac{c}{\beta}\right)^j \frac{x^\beta j}{a^j} \left(1 - (1-\epsilon)\frac{c}{\beta}\frac{x}{a}\right)^{a-j}\,dx.
\end{align}
Then
\begin{equation}
    \int_0^{\delta a^{\frac{1}{\beta}}} \binom{a}{j} (1+\epsilon)^j\left(\frac{c}{\beta}\right)^j \frac{x^\beta j}{a^j} \left(1 - (1-\epsilon)\frac{c}{\beta}\frac{x}{a}\right)^{a-j}\,dx
    = \int_0^\infty \phi_{a}(x)\,dx,
\end{equation}
where
\begin{equation}
    \phi_{a}(x) := \begin{cases}
    \binom{a}{j} (1+\epsilon)^j\left(\frac{c}{\beta}\right)^j \frac{x^\beta j}{a^j} \left(1 - (1-\epsilon)\frac{c}{\beta}\frac{x}{a}\right)^{a-j}\,dx &\quad \text{for }0\le x\le \delta a^{\frac{1}{\beta}}\\
    0&\quad \text{for }x>\delta a^{\frac{1}{\beta}}.
    \end{cases}
\end{equation}
We have that
\begin{equation}
\lim_{a\rightarrow \infty} \phi_{a}(x) = \frac{1}{j!}(1+\epsilon)^j\left(\frac{c}{\beta}\right)^jx^{\beta j}e^{-(1-\epsilon)\frac{c}{\beta}x^{\beta}}
\end{equation}
and
\begin{equation}
    \phi_{a}(x) \le \frac{1}{j!}(1+\epsilon)^j\left(\frac{c}{\beta}\right)^jx^{\beta j}e^{-(1-\epsilon)\frac{c}{\beta}x^{\beta}} (1 - (1-\epsilon)\frac{c}{\beta}\epsilon^\beta))^{-j} = C(j,\epsilon)x^{\beta j}e^{-(1-\epsilon)\frac{c}{\beta}x^{\beta}}
\end{equation}
for a constant $C(j,\epsilon)$ independent of $a.$ Now note that $\int_0^\infty x^{\beta j}e^{-(1-\epsilon)\frac{c}{\beta}x^{\beta}} < \infty.$ It follows from the dominated convergence theorem that
\begin{equation}
    \lim_{a\rightarrow \infty} \int_0^\infty \phi_{a}(x)\,dx
    = \int_0^\infty \lim_{a\rightarrow \infty} \phi_{a}(x)\,dx
    = \int_0^\infty \frac{1}{j!}(1+\epsilon)^j\left(\frac{c}{\beta}\right)^jx^{\beta j}e^{-(1-\epsilon)\frac{c}{\beta}x^{\beta}}\,dx < \infty.
\end{equation}
Therefore, for $a$ sufficiently large,
\begin{equation}
    \int_0^\delta \binom{a}{j} G_X(x)^j (1 - G_X(x))^{a-j}\,dx \le a^{-\frac{1}{\beta}}(1+\epsilon)\int_0^\infty \frac{1}{j!}(1+\epsilon)^j\left(\frac{c}{\beta}\right)^jx^{\beta j}e^{-(1-\epsilon)\frac{c}{\beta}x^{\beta}}\,dx.
\end{equation}
Analogously, we can show that for $a$ sufficiently large,
\begin{equation}
    \int_0^\delta \binom{a}{j} G_X(x)^j (1 - G_X(x))^{a-j}\,dx \ge a^{-\frac{1}{\beta}}(1-\epsilon)\int_0^\infty \frac{1}{j!}(1-\epsilon)^j\left(\frac{c}{\beta}\right)^jx^{\beta j}e^{-(1+\epsilon)\frac{c}{\beta}x^{\beta}}\,dx.
\end{equation}
Now observe that
\begin{align}
&\lim_{\epsilon\rightarrow 0^+} (1+\epsilon)\int_0^\infty \frac{1}{j!}(1+\epsilon)^j\left(\frac{c}{\beta}\right)^jx^{\beta j}e^{-(1-\epsilon)\frac{c}{\beta}x^{\beta}}\,dx\\
&= \int_0^\infty \frac{1}{j!}\left(\frac{c}{\beta}\right)^j x^{\beta j}e^{-\frac{c}{\beta}x^\beta}\,dx\\
&= \lim_{\epsilon\rightarrow 0^+} (1-\epsilon)\int_0^\infty \frac{1}{j!}(1-\epsilon)^j\left(\frac{c}{\beta}\right)^jx^{\beta j}e^{-(1+\epsilon)\frac{c}{\beta}x^{\beta}}\,dx,
\end{align}
where we once again apply the dominated convergence theorem. It follows that
\begin{equation}\label{eq:onion}
    \lim_{a\rightarrow \infty} \frac{\int_0^\delta \binom{a}{j} G_X(x)^j (1 - G_X(x))^{a-j}\,dx}{a^{-\frac{1}{\beta}}} = \int_0^\infty \frac{1}{j!}\left(\frac{c}{\beta}\right)^j x^{\beta j}e^{-\frac{c}{\beta}x^\beta}\,dx.
\end{equation}

\paragraph{The second integral.}
We now analyze
\begin{equation}
    a^{\frac{1}{\beta}}\int_\delta^\infty \binom{a}{j} G_X(x)^j (1 - G_X(x))^{a-j}\,dx.
\end{equation}
Observe that
\begin{align}
    a^{\frac{1}{\beta}}\int_\delta^\infty \binom{a}{j} G_X(x)^j (1 - G_X(x))^{a-j}\,dx
    &< a^{\frac{1}{\beta}} \binom{a}{j} \int_\delta^\infty (1 - G_X(x))^{a-j}\,dx\\
    &< a^{\frac{1}{\beta}} \binom{a}{j} (1 - G_X(\delta))^{a-j} \int_\delta^\infty 1 - G_X(x)\,dx\\
    &< a^{\frac{1}{\beta}} \binom{a}{j} (1 - G_X(\delta))^{a-j} \EE[X].
\end{align}
Thus,
\begin{equation}\label{eq:carrot}
    \lim_{a\rightarrow \infty} \frac{\int_\delta^\infty \binom{a}{j} G_X(x)^j (1 - G_X(x))^{a-j}\,dx}{a^{\frac{1}{\beta}}} = 0.
\end{equation}

Combining \eqref{eq:onion} and \eqref{eq:carrot} gives us that
\begin{equation}
\int_0^\infty \binom{a}{j} G_X(x)^j (1 - G_X(x))^{a-j}\,dx \propto a^{-\frac{1}{\beta}},
\end{equation}
as desired.

\paragraph{Proof of \Cref{thm:general}(iii).} Recall again that
\begin{equation}
    h(a) := \sum_{i=1}^{\min\{k, a\}} \mu_\DD(a-i+1,a).
\end{equation}
We show that $h$ is strictly concave and
\begin{equation}\label{eq:mouse-2}
\lim_{a\rightarrow \infty} h(a) - B\log a - C = 0
\end{equation}
for constants $B,C>0$. Both of these facts follow directly from the lemma below. \Cref{thm:general}(iii) then follows immediately by applying \Cref{lem:fennel}(iii).

\begin{lemma}\label{lem:mouse}
For $\DD$ an exponential distribution with rate parameter $\lambda$, so that $f_X(x) = \lambda e^{-\lambda x}$ for $\lambda > 0,$
\begin{equation}\label{eq:papaya}
    \lim_{a\rightarrow \infty} \mu_\DD(a-i,a) - \log a - B(j) = 0
\end{equation}
for a constant $B(j)>0.$ Moreover, $\mu_\DD(a-i,a)$ is strictly concave.
\end{lemma}

\begin{proof}
For an exponential distribution with rate parameter $\lambda,$ it is well known that
\begin{equation}\label{eq:bridge}
    \mu_\DD(a-i,a) = \sum_{j=i+1}^{a} \frac{1}{\lambda n}.
\end{equation}
It is clear, then, that $\mu_\DD(a-i,a)$ is strictly concave. \eqref{eq:bridge}  is equal to
\begin{equation}
    \frac{1}{\lambda}\left(\log n + \gamma + \epsilon(a) - \sum_{j=1}^i \frac{1}{j}\right),
\end{equation}
where $\gamma$ is the Euler-Mascheroni constant and $\lim_{a\rightarrow \infty}\epsilon(a) = 0,$ from which \eqref{eq:papaya} follows.
\end{proof}

\paragraph{Proof of \Cref{thm:general}(iv).}

Recall again that
\begin{equation}
    h(a) := \sum_{i=1}^{\min\{k, a\}} \mu_\DD(a-i+1,a).
\end{equation} Then it suffices to show that $h$ is strictly concave and
\begin{equation}
\lim_{a\rightarrow \infty} \frac{h(a)}{Ba^{\frac{1}{\alpha}}} = 1
\end{equation}
for a constant $B>0.$ Both of these facts follow directly from the lemma below. \Cref{thm:general}(iv) then follows immediately by applying \Cref{lem:fennel}(iv).

\begin{lemma}
For $\DD$ a Pareto distribution with pdf $f_X(x) = x^{-\alpha-1}$ for $\alpha > 1,$
\begin{equation}
    \lim_{a_t\rightarrow \infty} \frac{\mu_\DD(a-i,a)}{a^\frac{1}{\alpha}} = C
\end{equation}
for a constant $C>0.$ Moreover, $\mu_\DD(a-i,a)$ is strictly concave.
\end{lemma}

\begin{proof}
The result follows directly from Lemmas D.10 and D.11 in \cite{kleinberg2018selection}, where it is shown (in our notation) that
\begin{align}
    \lim_{a\rightarrow \infty}\frac{\mu_\DD(a,a)}{a^{\frac{1}{\alpha}}} = \Gamma\left(\frac{\alpha-1}{\alpha}\right)
\end{align}
and
\begin{equation}
    \mu_\DD(a-i,a) = \prod_{j=1}^i \left(1 - \frac{1}{j\alpha}\right)\mu_\DD(a,a).
\end{equation}
\end{proof}

Thus,
\begin{equation}
    \lim_{a\rightarrow \infty} \frac{\sum_{i=1}^{k} \mu_\DD(a-i+1,a)}{B \log a} = 1
\end{equation}
for a constant $B$. Also, note that $\mu_\DD(a-i,a)$ is a constant multiple of $\mu_\DD(a,a)$, and that $\mu_\DD(a,a)$ is strictly concave, since the mean of the largest order statistic of a distribution is strictly concave in sample size. Thus, $\mu_\DD(a-i,a)$ is strictly concave.

\subsection{Proof of \Cref{thm:ber-decay}}\label{sec:proof-thm2}

Suppose that for each fixed $i$, $X_i^{(t)}$ are i.i.d. Bernoulli random variables with success probability $q_i$ such that $q_i= c(i+d)^{-\beta}$ for some $\beta,c,d\ge 0.$ We begin by proving parts (i),(ii), and (iii), wherein the user only utilizes the highest value recommended item. In the Bernoulli setting, this amounts to determining the probability that \textit{at least one} movie is satisfactory, i.e.,
\begin{equation}
    \sum_{t=1}^m p_t h(a_t),
\end{equation}
where
\begin{align}\label{eq:grass}
    h(a) = 1 - \prod_{i=1}^{a} (1 - c(i+d)^{-\beta}).
\end{align}
It suffices now to show the desired asymptotic properties for $h$ depending on $\beta$, and applying \Cref{lem:fennel}.

We note the following fact, which will be helpful in our analysis:
\begin{equation}\label{eq:exp-bound}
    1-x>e^{-x-x^2}\quad \text{for $0< x < \frac{1}{2}$}.
\end{equation}

\subsubsection{Proof of \Cref{thm:ber-decay}(i)}
We now consider the case $0\le \beta < 1$. We analyze
\begin{equation}
    1 - h(a) = \prod_{i=1}^{a} (1 - c(i+d)^{-\beta}) = \prod_{i=d+1}^{a+d} (1 - ci^{-\beta}),
\end{equation}
and bound it from above and below. Let $i'$ be the smallest $i'$ such that $c(i')^{-\beta}<\frac{1}{2}$. Bounding from above,
\begin{align}
\prod_{i=i'+1}^{a+d} (1 - ci^{-\beta})
&< \prod_{i=i'+1}^{a+d} e^{-ci^{-\beta}}\\
&= \exp \left[\sum_{i=i'}^{a+d} -ci^{-\beta}\right]\\
&< \exp \left[\int_{i'}^{a+d+1} -cx^{-\beta}\,dx\right]\\
&= \exp \left[ -\left[ \frac{cx^{1-\beta}}{1-\beta} \right]_{i'}^{a+d+1}\right]\\
&= \exp \left[-\frac{c}{1-\beta} \left((a+d+1)^{1-\beta}-(i')^{1-\beta}\right)\right]
\end{align}
Now, bounding from below,
\begin{align}
\prod_{i=i'}^{a+d} (1 - ci^{-\beta})
&> \prod_{i=i'}^{a+d} e^{-ci^{-\beta}-c^2i^{-2\beta}}\\
&= \exp \left[\sum_{i=i'}^{a+d} -ci^{-\beta}-c^2i^{-2\beta}\right]\\
&> \exp \left[\int_{i'-1}^{a+d} -cx^{-\beta}-c^2x^{-2\beta}\,dx\right]\\
&= \exp \left[ -\left[ \frac{cx^{1-\beta}}{1-\beta} + \frac{c^2x^{2-\beta}}{2-\beta} \right]_{i'-1}^{a+d}\right]\\
&= \exp \left[-\frac{c}{1-\beta} \left((a+d)^{1-\beta}-(i'-1)^{1-\beta}\right) -\frac{c^2}{1-2\beta} \left((a+d)^{2-\beta}-(i'-1)^{2-\beta}\right)\right],
\end{align}
where the first inequality follows from \eqref{eq:exp-bound}. Then observing that
\begin{equation}
    1 - h(a) = \prod_{i=d+1}^{i'-1} (1 - ci^{-\beta})\cdot \prod_{i=i'}^{a+d} (1 - ci^{-\beta}),
\end{equation}
where the first product is constant in $a$, it follows that
\begin{equation}
    \lim_{a \rightarrow \infty} \frac{\log (1-h(a))}{-\frac{c}{1-\beta}a^{1-\beta}} = 1,
\end{equation}
as desired. The result follows from \Cref{lem:fennel}(i).

\subsubsection{Proof of \Cref{thm:ber-decay}(ii)}
We turn to the case $\beta=1$. We have that
\begin{align}
    1-h(a) = \prod_{i=d+1}^{a+d} \left(1 - \frac{c}{i}\right) &= \frac{d+1-c}{d}\cdot \frac{d+2-c}{d+1}\cdot \ldots \cdot \frac{a+d-c}{a+d},
\end{align}
which telescopes to
\begin{equation}
    \prod_{i=1}^c \frac{d+i-c}{a+d-i+1}.
\end{equation}
Now note that
\begin{equation}
    \lim_{a\rightarrow \infty} \frac{\prod_{i=1}^c \frac{d+i-c}{a+d-i+1}}{a^{-c}} = \prod_{i=1}^c (d+i-c)
\end{equation}
is a finite constant. Therefore,
\begin{equation}
    \lim_{a\rightarrow \infty} \frac{1 - h(a)}{Ba^{-c}} = 1
\end{equation}
for constant $B$.

\subsubsection{Proof of \Cref{thm:ber-decay}(iii)}
We now consider the case $\beta > 1$. In this case, we will again bound \eqref{eq:grass} from above and below. First note that
\begin{equation}
    1 - h(a) = \prod_{i=d+1}^{\infty} (1 - ci^{-\beta}) = S
\end{equation}
for a finite constant $S$.

We have that
\begin{align}
\prod_{i=a+d+1}^{\infty} (1 - ci^{-\beta})
&< \prod_{i=a+d+1}^{\infty} e^{-ci^{-\beta}}\\
&= \exp \left[\sum_{i=a+d+1}^{\infty} -ci^{-\beta}\right]\\
&< \exp \left[\int_{a+d+1}^{\infty} -cx^{-\beta}\,dx\right]\\
&= \exp \left[ -\left[ \frac{cx^{1-\beta}}{1-\beta} \right]_{a+d+1}^{\infty}\right]\\
&= \exp \left[-\frac{c}{1-\beta} (a+d+1)^{1-\beta}\right]
\end{align}
Therefore,
\begin{equation}
    \prod_{i=1}^{a+d} (1 - ci^{-\beta}) = \frac{S}{\prod_{i=a+d+1}^{\infty}} (1 - ci^{-\beta}) > S / \exp \left[-\frac{c}{1-\beta} (a+d+1)^{1-\beta}\right].
\end{equation}

Also,
\begin{align}
\prod_{i=a+d+1}^{\infty} (1 - ci^{-\beta})
&> \prod_{i=a+d+1}^{\infty} e^{-ci^{-\beta}-c^2i^{-2\beta}}\\
&= \exp \left[\sum_{i=a+d+1}^{\infty} -ci^{-\beta}-c^2i^{-2\beta}\right]\\
&< \exp \left[\int_{a+d}^{\infty} -cx^{-\beta}-c^2x^{-2\beta}\,dx\right]\\
&= \exp \left[ -\left[ \frac{cx^{1-\beta}}{1-\beta} + \frac{c^2x^{2-\beta}}{2-\beta} \right]_{a+d}^{\infty}\right]\\
&= \exp \left[-\frac{c}{1-\beta} (a+d)^{1-\beta} -\frac{c^2}{1-2\beta} (a+d)^{2-\beta}\right],
\end{align}
where the first inequality holds for $a$ sufficiently large due to \eqref{eq:exp-bound}.

Therefore,
\begin{equation}
    \prod_{i=1}^{a+d} (1 - ci^{-\beta}) = \frac{S}{\prod_{i=a+d+1}^{\infty}} (1 - ci^{-\beta}) < S / \exp \left[-\frac{c}{1-\beta} (a+d-1)^{1-\beta} -\frac{c^2}{1-2\beta} (a+d)^{2-\beta}\right].
\end{equation}

Now observe that
\begin{equation}
    \lim_{a\rightarrow \infty} -\frac{c}{1-\beta} (a+d+1)^{1-\beta} = 0
\end{equation}
and
\begin{equation}
    \lim_{a\rightarrow \infty} -\frac{c}{1-\beta} (a+d-1)^{1-\beta} -\frac{c^2}{1-2\beta} (a+d)^{2-\beta} = 0.
\end{equation}

Therefore,
\begin{equation}
    \lim_{a\rightarrow \infty} \frac{\prod_{i=1}^{a+d} (1 - ci^{-\beta})}{\frac{S}{1 - \frac{c}{1-\beta} (a+d+1)^{1-\beta}}} = 1.
\end{equation}
Also note that
\begin{equation}
    \frac{S}{1 - \frac{c}{1-\beta} (a+d+1)^{1-\beta}}
    = S - \frac{S \frac{c}{1-\beta}}{(a+d+1)^{\beta-1} + \frac{c}{1-\beta}}
\end{equation}

Also,
\begin{equation}
    \lim_{a\rightarrow \infty} \frac{\prod_{i=1}^{a+d} (1 - ci^{-\beta})}{ \frac{S}{1 - \frac{c}{1-\beta} (a+d-1)^{1-\beta} -\frac{c^2}{1-2\beta} (a+d)^{2-\beta}}} = 1
\end{equation}
and
\begin{equation}
    \lim_{a\rightarrow \infty} \frac{1 - \frac{c}{1-\beta} (a+d)^{1-\beta} -\frac{c^2}{1-2\beta} (a+d)^{2-\beta}}{1 - \frac{c}{1-\beta} (a+d)^{1-\beta}} = 1.
\end{equation}

It follows that
\begin{equation}
    \lim_{a\rightarrow \infty} \frac{1 - S - h(a)}{-\frac{Sc}{1-\beta}a^{1-\beta}} = 1
\end{equation}
as desired. The result follows from \Cref{lem:fennel}(ii).

\subsubsection{Proof of \Cref{thm:ber-decay}(iv)}
We now prove part (iv), where we consider the expected total value of all recommended items, which is equal to
\begin{equation}
    \sum_{t=1}^{m} p_t h(a_t),
\end{equation}
where
\begin{equation}
    h(a) = \sum_{i=1}^{a} c(i+d)^{-\beta}.
\end{equation}
Again, we consider three cases, $0< \beta < 1, \beta=1, \beta > 1.$

\paragraph{Case 1: $0 < \beta < 1$.} For $0 <\beta < 1,$ observe that
\begin{align}
    \sum_{i=1}^{a} c(i+d)^{-\beta} < \int_{d}^{a+d} cx^{-\beta}\,dx &= \left[\frac{c}{1-\beta}x^{1-\beta}\right]_{d}^{a+d}\\
    &= \frac{c}{1-\beta}(a+d)^{1-\beta} - \frac{c}{1-\beta}d^{1-\beta}
\end{align}
and
\begin{align}
    \sum_{i=1}^{a} c(i+d)^{-\beta} > \int_{d+1}^{a+d+1} cx^{-\beta}\,dx &= \left[\frac{c}{1-\beta}x^{1-\beta}\right]_{d+1}^{a+d+1}\\
    &= \frac{c}{1-\beta}(a+d+1)^{1-\beta} - \frac{c}{1-\beta}(d+1)^{1-\beta}.
\end{align} 
It follows that
\begin{align}
\lim_{a\rightarrow \infty}\frac{h(a)}{a^{1-\beta}} = \frac{c}{1-\beta},
\end{align} 
and the result in this case follows by applying \Cref{lem:fennel}(iv).

\paragraph{Case 2: $\beta = 1$.} Now for $\beta=1,$ we have that
\begin{align}
    \sum_{i=1}^{a} c(i+d)^{-\beta} = c\sum_{i=d+1}^{a+d} \frac{1}{i}.
\end{align}
\begin{align}
\lim_{a\rightarrow \infty}h(a) - c\log a + c\gamma - c\sum_{i=1}^d \frac{1}{i} = 0,
\end{align} 
where $\gamma$ is the Euler-Mascheroni constant The result in this case follows by applying \Cref{lem:fennel}(iii).

\paragraph{Case 3: $\beta > 1$.} Finally, for $\beta>1,$ we have that 
\begin{equation}
     \sum_{i=1}^{\infty} c(i+d)^{-\beta} = S
\end{equation}
for some finite $S$. Then note that
\begin{equation}
    \sum_{i=1}^{a} c(i+d)^{-\beta} = S - \sum_{i=a+1}^{\infty} c(i+d)^{-\beta}.
\end{equation}
Then we have
\begin{align}
    \sum_{i=a+1}^{\infty} c(i+d)^{-\beta} < \int_{a}^{\infty} cx^{-\beta}\,dx &= \left[\frac{c}{1-\beta}x^{1-\beta}\right]_{a}^{\infty}
    = - \frac{c}{1-\beta}a^{1-\beta}
\end{align}
and
\begin{align}
    \sum_{i=a+1}^{\infty} c(i+d)^{-\beta} > \int_{a+1}^{\infty} cx^{-\beta}\,dx &= \left[\frac{c}{1-\beta}x^{1-\beta}\right]_{a+1}^{\infty}
    = - \frac{c}{1-\beta}(a+1)^{1-\beta}
\end{align}
It follows that
\begin{equation}
    \lim_{a\rightarrow \infty}\frac{h(a)}{S - \frac{c}{\beta - 1}a^{1-\beta}} = 1.
\end{equation}
The result in this case follows again by applying \Cref{lem:fennel}(ii).

\subsection{A rounding lemma}\label{sec:proof-rounding-lemma}
The following lemma is useful for showing that---for the class of problems we consider here---optimal integer solutions are well-approximated by optimal real solutions.

\begin{lemma}
\label{lem:roundingloss}
Let $g_1,g_2,\cdots,g_m: [0,\infty)^m \rightarrow \RR$ be strictly convex functions over the non-negative reals. Then define
\begin{equation}
    g(x_1,\cdots,x_m):=\sum_{t=1}^m g_t(x_t).
\end{equation}
Then, under the constraint that $\sum_{t=1}^m x_t = n$, if $(x_1^*,\cdots,x_m^*)$ is the maximum of $g$ over the non-negative reals and $(a_1^*,\cdots,a_m^*)$ is the maximum of $g$ over the non-negative integers, then
\begin{equation}
\lfloor x_t^* \rfloor - m < a_t^* < \lfloor x_t^* \rfloor + m
\end{equation}
for all $t$.
\end{lemma}

\begin{proof}
The key idea is to show that there cannot be $i,j$ such that $a_i^*\ge \lceil x_i \rceil + 1$ and $a_j^*\le \lceil x_j \rceil - 1.$ If such a pair does exist, we show that
\begin{equation}
g(\cdots,a_i^*-1,\cdots,a_j^*+1,\cdots) \ge g(\cdots,a_i^*,\cdots,a_j^*,\cdots),
\end{equation}
contradicting the optimality of $a_1^*,\cdots,a_m^*.$ It suffices to show that
\begin{equation}
g_i(a_i^*-1) + g_j(a_j^*+1)\ge g_i(a_i^*) + g_j(a_j^*),
\end{equation}
or equivalently,
\begin{equation}
g_j(a_j^*+1) - g_j(a_j^*)\ge g_i(a_i^*) - g_i(a_i^*-1).
\end{equation}
This holds, as
\begin{align}
   g_j(a_j^*+1) - g_j(a_j^*) \ge \frac{\partial g_j}{\partial x_j} = \frac{\partial g_i}{\partial x_i} \ge g_i(a_i^*) - g_i(a_i^*-1).
\end{align}
\end{proof}

\subsection{Proof of \Cref{prop:uniform}}\label{sec:proof-of-uniform}
 For ease of exposition, we will begin by proving the result for $k=1$ and then handle the more general case.
\paragraph{First case ($k=1$).} Let us first consider the case $k=1.$ We would like to find an ordered tuple of non-negative integers $(a_1,\cdots,a_m)$ that maximizes
\begin{equation}
\sum_{t=1}^m p_t\mu_\DD(a_t,a_t)=\sum_{t=1}^m p_t\left(1 - \frac{1}{a_t+1}\right)
\end{equation}
subject to the constraint $\sum_{t=1}^m a_t = n.$
Our strategy will be to solve the relaxed optimization problem over non-negative reals, and then show that the optimal integer solution is ``close to'' the optimal real solution. Consider the function $g:[0,\infty)^m \rightarrow \RR$, where
\begin{equation}
g(x_1,\cdots,x_m) = \sum_{t=1}^m p_t\left(1 - \frac{1}{x_t+1}\right).
\end{equation}
Subject to the constraint $\sum_{t=1}^m x_t = n,$ $g(x_1,\cdots,x_m)$ is maximized exactly when 
\begin{equation}
\frac{\partial g}{\partial x_1} = \frac{\partial g}{\partial x_2} = \cdots = \frac{\partial g}{\partial x_m}.
\end{equation}
This is clear after noting that $g$ is convex on its domain and applying Lagrange multipliers. We have that
\begin{equation}
\frac{\partial g}{\partial x_t} = p_t\left(\frac{1}{x_t+1}\right)^2.
\end{equation}
So if $(x_1^*,\cdots,x_m^*)$ is a maximum, then $x_t^*+1 \propto \sqrt{p_t},$ meaning that
\begin{equation}
x_t^* = \left(\frac{\sqrt{p_t}}{\sum_{i=1}^m \sqrt{p_i}}\right)n-1.
\end{equation}
We now apply \Cref{lem:roundingloss}: if $(a_1^*,\cdots,a_m^*)$ is the optimal integer solution, then $|a_t^*-x_t^*|\le m$. Thus,
\begin{equation}
    \left|a_t^* - \frac{\sqrt{p_t}}{\sum_{i=1}^m \sqrt{p_i}}n\right|\le m+1
\end{equation}
for all $n$.

\paragraph{General case $k$.}  More generally, we would like to maximize
\begin{equation}
\sum_{t=1}^m p_t \sum_{i=1}^{\min\{a_t,k(n)\}} \mu_\DD(a_t-i+1, a_t).
\end{equation}
subject to the constraint $\sum_{t=1}^m a_t = n,$ where we let $k=k(n)$ be a function of $n$.
We first analyze each case of the inner summation:
\begin{equation}
\sum_{i=1}^{\min\{a_t,k(n)\}} \mu_\DD(a_t,a_t-i+1) =
\begin{cases}
\frac{a_t}{2} &\quad a_t\le k(n)\\
\sum_{i=1}^{k(n)} 1 - \frac{i}{a_t+1} = k(n) - \frac{k(n)^2 + k(n)}{2(a_t+1)}& \quad a_t> k(n)
\end{cases}.
\end{equation}
The top case $a_t \leq k(n)$ follows from all $a_t$ items of that type contributing to the objective, each with a mean value of $\frac12$.

Then define
\begin{equation}
g:[0,\infty)^m \rightarrow \RR,\quad
(x_1,\cdots,x_m)\mapsto \sum_{t=1}^m p_th(x_t)
\end{equation}
where
\begin{equation}
    h(x):=\begin{cases}
\frac{x}{2} &\quad x\le k(n)\\
\sum_{i=1}^{k(n)} 1 - \frac{i}{x+1} = k(n) - \frac{k(n)^2 + k(n)}{2(x+1)}& \quad x> k(n)
\end{cases}.
\end{equation}

So we have
\begin{equation}
\frac{\partial g}{\partial x_t} = 
\begin{cases}
\frac{p_tx}{2} &\quad x_t\le k(n)\\
p_t(k(n)^2 + k(n))\left(\frac{1}{x_t + 1}\right)^2 &\quad x_t> k(n)
\end{cases}.
\end{equation}
We first consider possible solutions where $x_t > k(n)$ for all $t\in[m].$ In this case,
\begin{equation}
\frac{\partial g}{\partial x_t} = p_t(k(n)^2 + k(n))\left(\frac{1}{x_t + 1}\right)^2
\end{equation}
is equal for all $t$ at the maximum, as before. Remarkably, the new $k(n)^2 + k(n)$ term drops out. So if $(x_1^*,\cdots,x_m^*)$ is the optimal real solution, again,
$x_t^*+1 \propto \sqrt{p_t},$ and the result follows as in the first case. Here, $x_t^* > k(n)$ holds whenever $k\le \frac{\sqrt{p_m}}{\sum_{i=1}^m \sqrt{p_i}}n - m - 1.$

\subsection{Proof of \Cref{thm:ber-varying}}\label{sec:proof-thm3}
\begin{proof}
We would like to find $a_1,\cdots,a_m$ that maximizes
\begin{equation}
\sum_{t=1}^m p_t \left(1 - (1 - q_t)^{a_t}\right) = 1 - \sum_{t=1}^m p_t(1-q_t)^{a_t}.
\end{equation}
subject to the constraint $\sum_{t=1}^m p_t = n.$ This is equivalent to minimizing
\begin{equation}
\sum_{t=1}^m p_t(1-q_t)^{a_t}.
\end{equation}
Now define a function
\begin{equation}
    g:[0,\infty)^m\rightarrow \RR,\quad (x_1,\cdots,x_m)\mapsto \sum_{t=1}^m p_t(1-q_t)^{x_t}.
\end{equation}
Subject to the constraint $\sum_{t=1}^m x_t = n,$ $g(x_1,\cdots,x_m)$ is maximized exactly when 
\begin{equation}
\frac{\partial g}{\partial x_1} = \frac{\partial g}{\partial x_2} = \cdots = \frac{\partial g}{\partial x_m}.
\end{equation}
We have
\begin{equation}
    \frac{\partial g}{\partial x_t} = -p_t(1-q_t)^{x_t}\log(1-q_t).
\end{equation}
Solving $\partial g/\partial x_i = \partial g/\partial x_j$ gives
\begin{align}
    p_i(1-q_i)^{x_i}\log(1-q_i) &= p_j(1-q_j)^{x_j}\log(1-q_j)\\
    \implies \log p_i + x_i\log(1-q_i) + \log\log(1-q_i) &= \log p_j + x_j\log(1-q_j) + \log\log(1-q_j)
\end{align}
It follows that $a_t\propto \frac{1}{\log(1-q_t)}$ for all $t$, where we have once again applied \Cref{lem:roundingloss}.
\end{proof}

\end{document}